\documentclass[11pt,a4]{article}
\usepackage{amsmath}
\usepackage{amssymb}
\usepackage{theorem}

\usepackage[affil-it]{authblk}

\usepackage{graphicx}
\usepackage{tikz}
\usepackage{subfigure}
\usetikzlibrary{positioning}
\usepackage{graphicx}
\usepackage{textcomp}

\theoremstyle{theorem}
\newtheorem{theorem}{Theorem}[section]
\newtheorem{proposition}[theorem]{Proposition}

\newtheorem{lemma}[theorem]{Lemma}
\newtheorem{corollary}[theorem]{Corollary}
{\theorembodyfont{\rmfamily}
  \newtheorem{example}[theorem]{Example}
 \newtheorem{remark}[theorem]{Remark}
}

\newenvironment{proof}{\noindent\textit{Proof.}}
{\QED\vskip\theorempostskipamount} 

\def\petitcarre{\vrule height4pt width 4pt depth0pt}
\def\QED{\relax\ifmmode\eqno{\hbox{\petitcarre}}\else{%
  \unskip\nobreak\hfil\penalty50\hskip2em\hbox{}\nobreak\hfil
  \petitcarre
  \parfillskip=0pt \finalhyphendemerits=0\par\smallskip}
  \fi}

\newcommand{\tr}[1]{\vphantom{#1}^t #1}

\DeclareMathOperator{\E}{\mathcal{E}}

\newcommand{\T}{\mathcal T}
\newcommand{\Z}{\mathbb Z}
\DeclareMathOperator{\Card}{Card}

\begin{document}

\title{Balancedness    and  coboundaries in symbolic  systems}

\author{Val\'erie Berth\'e 
\thanks{ This work was supported by the Agence Nationale de la Recherche  through the project  ``Dyna3S'' (ANR-13-BS02-0003).}}
\affil{ IRIF, CNRS UMR 8243, Universit\'e Paris Diderot -- Paris 7,
 Case 7014, 75205 Paris Cedex 13, France \\
berthe@irif.fr}

\author{Paulina Cecchi Bernales%
\thanks{The  second author was supported by the PhD grant CONICYT - PFCHA / Doctorado Nacional / 2015-21150544.}}
\affil{IRIF, CNRS UMR 8243, Universit\'e Paris Diderot -- Paris 7
 Case 7014, 75205 Paris Cedex 13, France\\
 Departamento de Matem\'atica y Ciencia de la Computaci\'on, Universidad de Santiago de Chile\\
  pcecchi@irif.fr }

\date{}

\maketitle

\begin{center}
{\em  This paper is dedicated to the memory of  our friend and colleague Maurice  Nivat, \\ for  his constant  and  encouraging support, and for all of his    ideas \\
which have nourished the present work.}
\end{center}

\begin{abstract}
This paper  studies  balancedness  for  infinite words  and  subshifts, both for letters and factors.  Balancedness  is a measure of disorder that amounts to  strong convergence properties  for frequencies. It measures  the difference between  the   numbers of occurrences of a   given word in   factors  of the same length.  We focus on  two families of words, namely dendric words and  words  generated by  substitutions.
 The family of dendric words includes Sturmian and    Arnoux-Rauzy words,  as well as  codings of regular  interval exchanges.  We prove that dendric words are balanced on letters if and only if they are balanced
on words.  In the substitutive case, we stress the role played by the  existence of coboundaries  taking  rational  values
  and show  simple criteria   when frequencies take rational values for   exhibiting imbalancedness. 
\end{abstract}


\section{Introduction} \label{sec:intro}

Balancedness for   words or   subshifts is a measure of disorder that provides strong convergence properties  for frequencies of   letters and   words.   In  ergodic terms, balancedness  can be  interpreted as an optimal speed of convergence of Birkhoff sums toward frequencies of words. 
In combinatorial terms,
given a finite alphabet ${\mathcal A}$, 
a  word $u\in  {\mathcal A}^{\mathbb Z}$ is said to be   {\em balanced }  on the  finite word  $v \in {\mathcal A}^*$  if there exists a constant $C_v$ such that for every pair  $(w,w')$  of   factors  of $u$  of the same length,   the difference between  the number of occurrences of 
$v$ in each      word  $w$ and $w'$   differs by at most $C_v$, that  is,  $||w|_v-|w'|_v|\leq C_v$, for $|w|=|w'|$,  where the notation  $|w|_v$  stands for the number of occurrences of 
$v$, and $ | w|$, for its length.

The  study of balancedness belongs to the  general domain of  aperiodic order (see e.g. \cite{BaakeGrimm}), and is  considered for words, as well as for tilings and  Delone sets in the context of quasicrystals, with  balancedness being   closely related to the notion of   bounded distance equivalence to  lattice   \cite{Lacz:92,HayKeKo:17,FreGar:18}.  Balancedness
  first occurred in  the form  of $1$-balance for  letters for infinite  words  defined over   a two-letter alphabet\footnote{ A binary  infinite words is  said   $1$-balanced if   the   difference between numbers of occurrences  of  a given  letter  in factors of the same length is bounded  by $1$.}  in the seminal papers by   Morse and Hedlund
 \cite{MorseHedlund:38,MorseHedlund:40} who  laid the groundwork  for   the study of symbolic dynamics and Sturmian words.  
The infinite words   that are $1$-balanced   over a two-letter alphabet are  indeed exactly  the   Sturmian words.

 The  notion of balancedness   was then  considered for larger alphabets, for $C$-balance, with $C>1$, and for  factors,  instead of letters. Words   over a 
   larger alphabet  that  are $1$-balanced  have been characterized in \cite{Hubert:00}  and shown  to   be closely related to Sturmian  words.    Moreover, Sturmian words   have been proved to be  balanced on their factors in \cite{FagnotVuillon:02}.
Note that the number of
1-balanced words of length $n$ is polynomial  \cite{Lipatov:82, Mignosi:91}  while the  number of $C$-balanced words of length
$n$ for $C>1$ is exponential  \cite{Lipatov:82} and, therefore, being 
$C$-balanced is relatively common.

This notion  is  natural and has thus  been widely studied
from many viewpoints, for instance in ergodic theory \cite{Cassaigne-Ferenczi-Zamboni:00,Cassaigne-Ferenczi-Messaoudi:08} to  prove absence of weak mixing properties, 
in number theory  in  connection with  Fraenkel's  conjecture \cite{Fraenkel:73,Tijdeman:2000} (see  also  \cite{Tijdeman:2000} and  the survey on balanced words \cite{Vuillon:03}), or else,  in operations research, for optimal routing  and scheduling; see e.g.  \cite{AltmanGH:00,BraunerCrama:04,BraunerJost:08,Tijdeman:80}. Balancedness  is   also  closely related to  symbolic discrepancy, as   
investigated in \cite{Adam:03,Adam:04}.  
 
 In the multidimensional framework, balancedness has been considered  both for    multidimensional words
 \cite{BerTij:02} and  for tilings  \cite{Sadun:16}.  This notion has to be compared with homogeneity  (related  to $0$-balance) such as  introduced by M. Nivat in  \cite{Nivat:02}. A   binary  two-dimensional word   $U$ in    $ \{0,1\}^{{\mathbb Z}^2}$  is $k$-homogeneous for a   finite subset  $F$ if,   whatever the position of the window $F$  in  $U$,
exactly
$k$
ones appear in the window. M. Nivat proved in  \cite{Nivat:02}  that a two-dimensional word  is  $1$-homogeneous  for a window  $F$ if and only if $F$  tiles the plane; recall that this latter property  has been     been characterized in \cite{BeauNiv:91}. 
Providing measures of   order for words or tilings   has constantly  been   in  M. Nivat's   research interests;  let us quote  e.g.  the paper   \cite{BrlekHNR:04} which
has  opened the way to the study of palindromic defects and complexity,   and also  Nivat's conjecture   which states that if a two-dimensional word admits  at most $ m n$ rectangular factors of size $(m,n)$, then it admits at least one direction of periodicity. This  elegant and  apparently simple conjecture  has been formulated    by M. Nivat in 1997  during an invited talk at ICALP and has  lead to various approaches and    numerous   results.
See  for instance  \cite{CyrKra:2015,KariMoutot:18,kariSza:2015}  for  further  references and   examples of  the latests developments on  this  conjecture.

In the case of   primitive substitutions,   if   balancedness  is   known to be closely related    to the  eigenvalues of the substitution matrix (see \cite{HolZam:98,Adam:03,Adam:04}  and Section~\ref{subsec:subs}),   it is   not completely characterized
in purely linear terms (see  Remark \ref{rem:adam}).   Note that there is an important literature devoted to   the measure of imbalancedness, considered as deviations  for Birkhoff sums,
for substitutions  acting  on words or, in the higher-dimensional framework, for tiling substitutions.  
See  in particular \cite{BressaudBH,paquette_son_2017} where  central limit theorems are considered  for deviations of ergodic sums   for substitution subshifts. In \cite{BufSol:2013},
  deviation of ergodic averages for  ${\mathbb R}^d$-actions by translation  associated with self-similar tilings are  also investigated with new phenomena  due to the  dimension; this is the so-called boundary effect  expressed in terms of the eigenvalues of the underlying substitution matrix. See also \cite{Sadun:2011} based  on the  use of cohomology for tiling spaces. 
Similarly, a  criterion for balancedness  for $S$-adic words (obtained by iterating several substitutions and not just one)  is given  in  \cite[Theorem 5.8]{Berthe-Delecroix}
in terms of convergence of  products of matrices; 
see  \cite{Delecroix-Hejda-Steiner,BerST:18} for an  example of application for Brun substitutions.   We  discuss further  $S$-adic  examples in Section~\ref{subsec:exdendric}. 
 \bigskip

In  this paper, we     study  balancedness  for  some families of   words of low  factor complexity\footnote{Factor complexity counts the number of factors of a given size.},  not only for letters, but also  for factors.  We  focus on two   families of words, namely   dendric words in Section \ref{sec:dendric} and  fixed points of  primitive substitutions
having rational frequencies   in Section \ref{sec:sub}.  Dendric words   are also called  tree words (see e.g.~\cite{BDFDLPRR:15,BDFDLPRR:15bis,BFFLPR:2015,BDFDLPR:16,Rigidity:17}).  These two families behave  in  different ways with respect to balancedness. 
Frequencies  of  factors of dendric words are irrational  \cite{Rigidity:17} and  balancedness for letters  is equivalent to   balancedness  for factors (Theorem~\ref{theo:equilibre}).   In the case of  substitutions, the example of  the Thue--Morse substitution (handled in  Corollary \ref{cor:TM})  illustrates the fact that one can have balancedness on letters  
and   imbalancedness on factors. Our approach does not rely on  techniques of linear algebra using  the   substitution matrix, or, in the dendric case, on the underlying   $S$-adic  expansion.
It also  allows us  to  prove    balancedness results for specific  factors. 
We  exploit  ideas  issued from topological dynamics  (and in particular,  the notion of coboundary)  for the study  of   balancedness  of   fixed points of primitive substitutions  words  with rational frequencies in Section \ref{sec:sub}  and, conversely,  we 
use balancedness to deduce spectral properties in Section \ref{sec:dendric},   which uses purely  combinatorial methods. 
Our  main results are the following. 
\begin{theorem}\label{theo:equilibre}
Let $(X,T)$ be a minimal dendric subshift. Then $(X,T)$ is balanced on  letters if and only if it is  balanced  on  factors.
In particular, if $(X,T)$  is balanced, then  all the   frequencies of   factors are additive topological eigenvalues
and all 
cylinders are  bounded remainder sets.
\end{theorem}
\begin{theorem}\label{condnec1}
Let  $\sigma$ be a  primitive substitution over the alphabet ${\mathcal A}$.  Let $\mathcal{L}(X_\sigma)$ denote the language of $\sigma$. Let 
$v$ be in $ \mathcal{L}(X\sigma)$, and  suppose that it has  a  rational frequency  $\mu_v=p_v/q_v$ written in  irreducible form.   Suppose that   the associated subshift  $(X_{\sigma},T)$ is balanced on $v$.  Then, we have the following.
\begin{enumerate}
\item 
For each $a\in\mathcal{A}$ and   each return  word $w$ to $a$, $q_v$ divides $|\sigma^n(w)|$ for  all $n$ large.  In particular, if $aa\in \mathcal{L}_2(X_{\sigma})$, then $q_v$ divides $|\sigma^n(a)|$ for all $n$ large.

\item 
Let $a\in \mathcal{A}$ and suppose that 
there exist $b,c\in \mathcal{A}$ such that $bac\in \mathcal{L}(X_{\sigma})$ and $bc\in \mathcal{L}(X_{\sigma})$. 
Then $q_v$ divides $|\sigma^n(a)|$ for  all  $n$ large.
\end{enumerate}
\end{theorem}

We briefly describe  the contents of this paper.  Basic  definitions are recalled in Section~\ref{sec:def}.  In particular, we  stress the relations between
balancedness and discrepancy in Section~\ref{subsec:dis},   we recall   balancedness results  for  substitutions in Section \ref{subsec:subs},
and highlight the connections with   coboundaries and   spectral eigenvalues in Section \ref{subsec:coboundaries}. The approach of Section \ref{sec:dendric} is  combinatorial and applies to a wide  family of infinite words and subshifts,  namely the  family of  dendric words.  The case of Arnoux-Rauzy words is discussed in Section 
\ref{subsec:exdendric}.
In Section \ref{sec:sub}, a topological  dynamics approach is developed for  infinite words for which frequencies do exist and are  rational,  based on  the existence of   coboundaries taking  rational values. Some examples are discussed in Section \ref{subsec:examples}.

\bigskip

{\bf Acknowledgements}
We would like to  thank J. Cassaigne, V. Delecroix, F. Dolce, F. Durand,  J. Leroy, D. Perrin, S. Petite and R. Yassawi  for  many  stimulating discussions  on the subject.
We also thank   the referee for her/his   comments  that   contributed to improving  the  readability and  quality of this paper. 

\section{Basic definitions} \label{sec:def}

\subsection{Words and   symbolic dynamical systems} 
 Let ${\mathcal A}$ be a finite non-empty alphabet of cardinality $d$.  
Let us denote by $\varepsilon$ the empty word of the free monoid ${\mathcal A}^*$, and by ${\mathcal A}^{\mathbb{Z}}$ the set of bi-infinite words over ${\mathcal A}$. For $a \in  {\mathcal A}$ and for $w \in {\mathcal A}^*$, $|w|_a$ stands for the number of  occurrences of the letter $a$ in the word $w$, and $|w|$ stands for the length of $w$. The  $i$th letter of $w$  is  denoted as  $w_i$  by labelling  indices  from $0$, i.e.,
$w=w_0 \cdots w_{|w|-1}$.  The notation $w_{[i,j]}$ stands for  the word $w_i \cdots w_j$,  for $i,j$ non-negative integers with $ i \leq j$, and $w_{[i,j)}$ stands for  $w_i \cdots w_{j-1}$,  with $i <j$. 
A~\emph{factor} of a (finite or infinite) word~$u$ is defined as the concatenation of consecutive letters occurring in~$u$. We use the  notation $w\prec u$  for  $w$  a factor of $u$.
The set of factors   ${\mathcal L }(u)$ of an infinite word  $u$ is called its {\em language}.  The {\em factor complexity} of  the  bi-infinite word $u$  counts the number of factors of  $u$ a  given length. 
The bi-infinite word 
  $u=(u_n)_n \in  {\mathcal A} ^{{\mathbb Z}}$ is  said  to be 
  {\em uniformly recurrent}
if every word occurring in $u$ occurs in an infinite number of 
positions  with
   bounded gaps, that is, 
for every factor $w$, there exists $s$ such
  that,  for every $n$, $w$ is a factor
of $u_n \dots u_{n+s-1}$.   Let $u$ be a uniformly recurrent  bi-infinite word  in   ${\mathcal A}^{\mathbb Z}$. Let $v\in\mathcal{A}^*$
be a factor of $u$. We say that a word $w$  with $wv$ in ${\mathcal L} (u)$  is a {\it first return word} to $v$ in $u$ if $v$ is a prefix of $wv$ and there are exactly two occurrences of $v$ in $wv$. We say that it is just a {\it return word} to $v$ in $u$ if $v$ is a prefix of $wv$.

\medskip

{\bf Symbolic dynamical systems.} Let $T$  stand for the {\em  shift}  acting on ${\cal A}^{{\mathbb Z}}$, that is, 
$T((u_n)_{n \in \mathbb{Z}})=(u_{n+1})_{n \in \mathbb{Z}}$.    A {\em subshift}
(also called a {\em shift)}
is a pair $(X,T)$ where $X$ is a closed shift-invariant subset of some ${\mathcal A}^{\mathbb Z}$. Here the set ${\cal A}^{{\mathbb Z}}$ is  equipped
with the product topology of 
the discrete
   topology on each copy of ${\cal  A}$.    
One  associates with any bi-infinite word  $u$ in ${\cal A}^{{\mathbb Z}}$ the  symbolic dynamical system   $( X_u,T)$,  where the  subshift  $X_u \subset {\mathcal A}^{\mathbb Z}$  is defined as  
$X_u=\{v\in \mathcal{A}^\mathbb{Z}: \forall w, w \prec v \Rightarrow  w\prec u \}.
$ A subshift $X$  is said to be {\em minimal} if it admits no non-trivial closed and shift-invariant subset.
 This is equivalent to the fact that  every bi-infinite word  $u$ in $X$ is uniformly recurrent.
If $X$ is a subshift, then its {\em  language} ${\mathcal L}(X)$ is defined as the set of factors of elements of $X$.
For any $n\geq 1$,  we let denote by  ${\mathcal L}_n(X)$ the set of  factors  of length $n$  of elements in  $X$.

\medskip

{\bf Substitution dynamical systems.} A  {\em substitution} $\sigma$  defined on the alphabet ${\mathcal A}$  is  a non-erasing morphism of the  free monoid ${\mathcal A}^*$, i.e., there is no letter in ${\mathcal A} $ whose image
under $\sigma$ is the empty word. If  there exists  a letter $a$ such that $\sigma(a)$ admits  the letter $a$ as  a strict prefix, then 
there exists an infinite word $u= \sigma^{\omega} (a)$ such that $\sigma(u)=u$. 
Moreover,  if $\sigma(b)$  admits  the letter $b$ as  a strict suffix for some letter $b$,  then there  exists
a bi-infinite word $v= \sigma^{\omega}(b)\cdot \sigma^{\omega} (a)$ such that $\sigma(v)=v$,  where the dot is located between the    letters of index $-1$ and $0$.  Such  infinite and bi-infinite words are said to be {\em fixed points}
of the substitution $\sigma$.

  Let $|\mathcal{A}|$ stand  for the cardinality of $\mathcal{A}$.
The {\em substitution matrix} (or incidence matrix)  $M_\sigma$ of $\sigma$ is the $|\mathcal{A}|\times|\mathcal{A}|$-matrix whose coefficients are
$M_\sigma(a,b)=|\sigma(b)|_a$. 
The substitution $\sigma$ is said to be {\it primitive} if there exists a power of $M_\sigma$ which is positive.
Given a  primitive substitution $\sigma$ on the finite alphabet $\mathcal{A}$, the symbolic system associated with $\sigma$ is the pair $(X_\sigma,T)$, where
$$X_\sigma:=\{x\in \mathcal{A}^{\mathbb Z}: \forall w, w \prec x \Rightarrow \exists a\in \mathcal{A}, \exists n\in \mathbb{N} \mbox{ s.t. } w\prec \sigma^n(a) \}.$$
Note that the subshift $X_{\sigma}$  is generated by any   bi-infinite word $v$ such that $\sigma^k(v)=v$  for some positive $k$, i.e.,
$X_{\sigma}= X_v$.
The {\em language} of $\sigma$ is  defined as ${\mathcal L}(X_{\sigma})$. Primitive substitutions  are known to be \emph{recognizable},  i.e.,  they can be  uniquely desubstituted \cite{Mosse:92,Mosse:96}.
More precisely, for any $x \in X_{ \sigma}$, there exists a  unique pair $(y,k)$ with  $y  \in  X_{ \sigma}$
and $0 \leq k < |\sigma(y_0)|$ such  that
$x = T^k \sigma(y)$. 

According  to Perron's theorem, if   a substitution  is primitive, then its  substitution matrix 
 admits a dominant eigenvalue  (it dominates  strictly  in modulus the other eigenvalues) that is positive. It is called its  {\em Perron eigenvalue},
 or its {\em expansion} factor.
 {\em  Pisot substitutions} are primitive substitutions  such that the    dominant Perron eigenvalue  of their substitution matrix is  a  {\em Pisot number}, that is, an  algebraic integer whose conjugates 
lie  strictly inside  the unit disk. 
Note that Pisot  irreducible  substitutions, that is,  Pisot  substitutions  for which the characteristic polynomial   of their substitution  matrix is  irreducible, are conjectured to   have pure discrete spectrum in the measure-theoretic  sense.  This  is called the Pisot substitution  conjecture
(see e.g. the survey \cite{AkiBBLS}).

\medskip

{\bf Frequencies and invariant measures.} Let $u$ be an bi-infinite word in ${\mathcal A}^{\mathbb Z}$.  The \emph{frequency}  $\mu_v$ of a   finite word $v \in {\mathcal A}^* $   is defined as the limit, when
$n$ tends toward infinity, if it exists, of the number of occurrences of $v$ in $ u_{-n}  \cdots u_{-1} u_0 u_1 \cdots u_{n-1}$ divided
by $2n+1$, i.e.,  
 $$\mu_v= \lim_{ n \rightarrow +\infty} \frac{| u_{ -n} \cdots u_0 \cdots u_n|_v}{2n+1}.$$
 Assume that   the frequencies of  the factors  of $u$  all exist.
 The infinite word $u$  is said to have  \emph{uniform frequencies} if, for every  factor $v$ of $u$, the convergence toward $\mu_v$  of 
 $ \frac{| u_{ k}  \cdots u_{k+2n} |_v}{2n+1} $   is   uniform  in $k$, when $n$ tends to infinity.

Let $(X,T)$ be a subshift  with  $X  \subset {\mathcal A }^{\mathbb Z}$. 
 A probability measure $\mu$ on $X$ is
said to be {\em $T$-invariant} if $\mu(T^{-1} A) = \mu(A)$ for every measurable set $A \subset X$.
The subshift $(X,T)$ is  {\em uniquely ergodic} if  there
exists a unique shift-invariant probability measure on $X$.   The subshift $(X,T)$ is uniquely ergodic if and only if  every    bi-infinite  word $u$  in $X$  has uniform  factor frequencies.  In that case, one
recovers the frequency  $\mu_v$ of a factor $v=v_0 \cdots v_n$   as $\mu_v= \mu([v])$, where the {\em cylinder}
$$[v] : = \{u \in X; u_0 \ldots u_{n} = v\}.$$  For more  on invariant measures and ergodicity, we refer to \cite{Queffelec:2010} and \cite[Chap. 7]{CANT}. We  also  recall  that  the  substitutive subshift $X_\sigma$ determined by a primitive substitution $\sigma$ is  uniquely
ergodic  and minimal (see e.g.  \cite{Queffelec:2010}).

\subsection{Balancedness and discrepancy }\label{subsec:dis}

A  bi-infinite word $u\in  {\mathcal A}^{\mathbb Z}$ is said to be   {\em balanced }  on the   factor   $v \in {\mathcal L}(u)$  if there exists a constant $C_v$ such that for every pair  $(w,w')$  of   factors  of $u$,   if $|w|=|w'|$,  then $$||w|_v-|w'|_v|\leq C_v.$$  It is   {\em  balanced on  letters} if it is balanced on each letter  in $\mathcal{A}$,    it is  {\em balanced on   factors}   if
 it is balanced on all  its   factors, and lastly, it  is  {\em balanced  on   factors  of length $n$}  if it is balanced on  all its factors of length $n$. 
 Similarly,  a subshift  $(X,T)$ with $X \subset {\mathcal A}^{\mathbb Z}$  is said to be   {\em balanced }  on the factor $v \in {\mathcal L} (X)$  if there exists a constant $C_v$ such that for all $w,w'$ in $\mathcal{L}(X)$ with $|w|=|w'|$, then  $||w|_v-|w'|_v|\leq C_v$.  The notions of  balancedness for  letters, words or words of a given length extend   similarly to subshifts. Note  that in  the first papers  devoted  to balancedness,
   balance   was used to  refer to $1$-balance  for letters (see e.g. \cite{Lothaire:2002}).

Proposition \ref{prop:decrease} below  (which is  a  rephrasing of  \cite[Lemma 23]{Adam:03}) 
 states   that   balancedness is preserved when  decreasing   the length of factors.
It is thus  sufficient to  prove   that  balancedness does not hold for some length to obtain that  it does not hold for all   larger lengths.
This will  be used in the examples of Section \ref{subsec:examples}. 

\begin{proposition}\cite[Lemma 23]{Adam:03} \label{prop:decrease}
If  a bi-infinite word    $u$  is   balanced  on  some factor $v$,  then it is balanced    on   the prefix of $v$ of length $|v|-1$.
If  a  bi-infinite word  $u$ or a subshift $(X,T)$  is  balanced on factors of length $n+1$, then  it is  balanced on factors of length $n$. 
\end{proposition}

\begin{proof}   Let $u \in {\mathcal A}^{\mathbb Z}$.
 For every $n$,  we  consider an alphabet  ${\mathcal A}_n$   and  a  bijection   $  \theta_n \colon  {\mathcal A}_n \rightarrow  {\mathcal L}_n(u)$.
 The word $u^{(n)}:= \theta_n(u)$, defined over the alphabet  ${\mathcal A}_n$,  codes 
factors of length $n$  according  to the bijection $\theta_n$ in  the same order as in $u$  with overlaps  and   without repetition.
The map  $\theta_n  \circ \pi_n  \circ  \theta_{ n+1} ^{-1}$ is a morphism  from the monoid  ${\mathcal A}_{n+1}^*$ to 
${\mathcal A}_{n+1}^*$  that maps  letters to letters: it maps  the coding of  a  block of length $n+1$ to   the coding of its prefix of length $n$. 
The word $u^{(n)}$   is thus  the image by a  letter-to-letter  substitution of the   word $u^{(n+1)}$; indeed 
$u^{(n)} = \theta_n  \circ \pi_n  \circ  \theta_{ n+1} ^{-1} (u^{(n+1)}).$
We conclude  by noticing that the action of a letter-to-letter substitution preserves balancedness. 
\end{proof}
\begin{example}\label{ex:TMcode}
We consider the Thue--Morse substitution $\sigma_{TM}$ defined over $\{0,1\}$ as  $\sigma_{TM} \colon 0 \mapsto 01, \  1 \mapsto 10.$ One has ${\mathcal L}_2 (\sigma_{TM})= \{00,01,10,11\}$.
Let 
$a=00$, $b=01 $, $c=10$, $d=11$.
The  central  letters of the two-sided  word  $u=(\sigma^2)^{\omega}( 0)\cdot (\sigma^2)^{\omega} (0)$ are
$$ \cdots 0110  1001  1001 0110   \cdot 0110  1001 1001 0110 \cdots $$
coded by  the central letters of $u^{(2)}$ 
$$ \cdots bdcb     cabd   cabc    bdca \cdot bdcb  cabd  cabc   bdcb \cdots $$
\end{example}

\medskip

{\bf Discrepancy.} Let $u \in {\mathcal A}^{\mathbb Z}$ be a bi-infinite word and assume that each  factor  $v \in {\mathcal L}(u)$ admits a  frequency $\mu_v$ in $u$. The \emph{  discrepancy} $\Delta_v(u) $ of $u$  with respect to $v$ is
defined as 
\[
\Delta_v(u) =\sup _{n \in \mathbb{N}} | | u_{-n}  \cdots  u_0 \cdots u_{n}| _v -  (2n+1)  \mu_v|.  
\]
The quantity  $\Delta_v(u)$ is considered e.g. in \cite{Adam:03,Adam:04}.
One easily  checks that   $u$ is   balanced on the factor $v$  if and only if its discrepancy $\Delta_v(u)$ is finite.
If  $\Delta_v(u)$ is finite, then the cylinder    $[v]$ is  said to be a    {\em bounded remainder set}, according to the terminology developed in  classical  discrepancy  theory. 
These definitions extend to any  subshift $(X,T)$ in a  straightforward way.

\subsection{Balancedness and substitutions} \label{subsec:subs}

We now consider the special case of substitutions. We follow  here \cite{Queffelec:2010}. Let $\sigma$ be a primitive substitution. Letter frequencies for a primitive  substitution  $\sigma$  are known  to  be provided by the  normalized  positive eigenvector  (whose sum of coordinates equal $1$) associated with the Perron 
 eigenvalue  of the substitution matrix $M_{\sigma}$.  We call it the {\em renormalized  Perron eigenvector}. In other words, 
 the frequency $\mu_a$ of  the   letter $a$   is the  $a$th entry of this vector.  Similarly,  frequencies of   factors  in ${\mathcal L} (X_{\sigma})$ (not only letters)
 are  known to be provided by  eigenvectors of   a  matrix,  namely the two-letter  substitution matrix  that is  associated  with $\sigma$ as follows. This  construction relies on the  coding by  $k$-letter blocks  
 described in the proof of Proposition \ref{prop:decrease}.

\medskip

  {\bf Two-letter factor  substitutions.} 
Given a substitution $\sigma$, consider the finite set $\mathcal{L}_2(X_\sigma)$ as an alphabet and define the 
{\em two-letter factor substitution} $\sigma_2$ on $\mathcal{L}_2(X_\sigma)$ as follows (it is also called {\em induced substitution} in \cite{BaakeGrimm}): for every $u=ab\in \mathcal{L}_2(X_\sigma)$, $\sigma_2(u)$ is the word  over the alphabet  $\mathcal{L}_2(X_\sigma)$ made of the first $|\sigma(a)|$ factors of length 2 in $\sigma(u)$. For instance, if $ab\in \mathcal{L}_2(X)$ with $\sigma(a)=a_0\cdots a_r$, $\sigma(b)=b_0\cdots b_s$, then
$$\sigma_2(ab)=(a_0a_1)(a_1a_2)\cdots (a_{r-1}a_r)(a_rb_0).$$
We recall from   \cite[Lemma 5.3--5.4]{Queffelec:2010} that if the substitution $\sigma$ is primitive, then $\sigma _2$ is also primitive, and $\sigma_2$ has the same 
Perron eigenvalue as $\sigma$.
 Frequencies of factors of length $2$ are  provided by the  renormalized Perron eigenvector of $M_{\sigma_2}$.

One then  checks that the  substitution  $\sigma_2$ admits  a  bi-infinite fixed  point that is composed   by all the factors of length  $2$ of  $v$  without repetition and in the same order as in $v$.

\begin{example}\label{ex:TMsigma2}
We continue Example \ref{ex:TMcode}
and  consider the Thue--Morse substitution $\sigma_{TM}$ defined over $\{0,1\}$ as  $\sigma_{TM} \colon 0 \mapsto 01$ and $\sigma_{TM} \colon 1 \mapsto 10.$ One has ${\mathcal L}_2 (\sigma_{TM})= \{00,01,10,11\}$.
One  has in particular
$\sigma(00)= 0101$ and  $\sigma_2(00)= (01)(10)$.
One  checks that  $\sigma^{(2)} (a)=bc$,  $\sigma^{(2)} (b)=bd$,  $\sigma^{(2)} (c)=ca$, $\sigma^{(2)} (d)=cb$,  by setting
$a=00$, $b=01 $, $c=10$, $d=11$. Observe that the image of $0$ by $\sigma^2$ begins and ends  with $0$. We  thus can  consider the  bi-infinite word 
$(\sigma^2)^{\omega}( 0)\cdot (\sigma^2)^{\omega} (0)$. Similarly,   the image of $b$ by $\sigma_2^2$ begins  with $b$, 
the image of $a$ by $\sigma_2^2$ ends  with $a$,
and we can consider  the  bi-infinite word 
$(\sigma_2^2)^{\infty}( a)\cdot (\sigma_2^2)^{\infty} (b)$.
Note that  powers of  the substitution $\sigma$  generate the same  subshift as $\sigma$.
The  central  letters of the two-sided  word  $(\sigma^2)^{\omega}( 0)\cdot (\sigma^2)^{\omega} (0)$ are
$$ \cdots 0110  1001  1001 0110   \cdot 0110 \, 1001 \, 1001\,  0110 \cdots $$
whose coding is provided by 
the  central  letters of the two-sided  word  $(\sigma_2^2)^{\omega}  (a)\cdot (\sigma_2^2)^{\omega} (b)$ 
$$ \cdots bdcb \,    cabd \,  cabc \,   bdca \cdot bdcb \, cabd \, cabc \,  bdcb \cdots $$
The  eigenvalues of $M_{\sigma}$ are  $2$ and $0$, and the eigenvalues of $M_{\sigma_2}$  are $0$, $1$, $-1$  and $2$.
\end{example}

 We can similarly define a notion of  higher-order   factor  substitution $\sigma_k$   with respect to factors of length $k$. 
 From $M_{\sigma}$ to  $M_{\sigma_2}$,  eigenvalues of modulus $1$ can be added, such as illustrated by  the example of the Thue--Morse substitution (see Example~\ref{ex:TMsigma2}). 
One crucial property is that 
the set of  eigenvalues of $M_{\sigma_2}$ is   a subset  of the  spectrum of $M_{\sigma_k}$ for $ k \geq 2$; only the eigenvalue $0$ is added, according to 
\cite[Corollary 5.5]{Queffelec:2010}.

\begin{example}\label{ex:TMsigma3}
We continue Example \ref{ex:TMcode} and \ref{ex:TMsigma2} with  the Thue--Morse substitution $\sigma_{TM}$.
One  has ${\mathcal L}_3 (\sigma_{TM})= \{001,010, 011,100,101,110\}$. One checks that 
 $\sigma^{(3)} (a)=be$,  $\sigma^{(3)} (b)=cf$,  $\sigma^{(3)} (c)=cf$,  $\sigma^{(3)} (d)=da$,  $\sigma^{(3)} (e)=da$,  $\sigma^{(3)} (f)=eb$, by setting $a=001$, $b=010$, $c=011$, $d=100$, $e=101$, $f=110$.
The  eigenvalues of  $M_{\sigma_3}$  are $0$, $1$, $-1$  and $2$, and are the  same as for $M_{\sigma_2}$; the eigenvalue $0$ has  multiplicity  $3$ for  $M_{\sigma_3}$
and $1$ for  $M_{\sigma_2}$,  and the other eigenvalues have the same multiplicity for  both matrices.
\end{example}

Moreover, we can compute the frequency of any factor thanks to  $M_{\sigma}$ and $M_{\sigma_2}$.    Roughly, given  a factor, it occurs in the image  by some    power $\sigma^n$ of $\sigma$  of   either  a letter or of a factor of length $2$. It  thus can be determined by
the  normalized Perron eigenvectors  of $M_{\sigma}$ and $M_{\sigma_2}$.
As a consequence,   balancedness for  factors can  be described  in terms of  the spectrum of the substitution matrices
$M_{\sigma}$ and $M_{\sigma_2}$,  which  even  provide estimates on the symbolic discrepancy and on the balance  function $B_n(u)$ \cite{Adam:03,Adam:04}.

\begin{theorem}[\cite{Adam:03,Adam:04}]\label{theo:Adam}
Let $\sigma$ be a primitive substitution. If $\sigma $  (resp. ${\sigma_2}$)  is a  Pisot substitution, 
  then  the subshift $X_{\sigma}$ is  balanced  on letters (resp.  on factors).
Conversely, if   $X_{\sigma}$ is  balanced  on letters (resp. on factors), then  the  Perron eigenvalue  of   $M_{\sigma}$   (resp.  $M_{\sigma_2}$) is the unique eigenvalue  of $M_{\sigma}$   (resp.  $M_{\sigma_2}$) that is  larger than  $1$ in modulus, and  all possible  eigenvalues of modulus one  of $M_{\sigma}$  (resp.  $M_{\sigma_2}$) are  roots of unity.
\end{theorem}

\begin{remark}\label{rem:adam}
In the case where  the matrix $M_{\sigma}$ admits a root of unity as an eigenvalue,     we cannot decide whether balancedness  holds or not just  by  inspecting  the matrix $M_{\sigma}$. An example of a  primitive and aperiodic  substitution  $\sigma$ that is balanced on  factors is given in Example~\ref{ex:theta2} having 
$1$  as an eigenvalue (aperiodic means that   $X_{\sigma}$  contains no periodic word). 
Note also that a  substitution can be balanced on letters but    unbalanced on factors if there exists an eigenvalue  of   modulus $1$  that occurs 
in the spectrum  of $M_{\sigma_2}$. This is the case of the Thue--Morse substitution   (see Example  \ref{ex:TMsigma2} and Corollary \ref{cor:TM}).
We   also deduce from Theorem~\ref{theo:Adam}  that 
primitive Pisot substitutions     have finite   discrepancy with respect to letters.
\end{remark}

      \begin{example} \label{ex:theta2} This example has been  communicated by J. Cassaigne and M. Minervino.
We consider the substitution $\sigma$ over the alphabet $\{1,2,3\}$ defined by 
   $ \sigma \colon 1 \mapsto 121, 2 \mapsto 32, 3 \mapsto 321$.
   Its spectrum is $\{1, \frac{3 \pm \sqrt 5 }{2}\}$ and it is  balanced on factors. 
   This  thus provides  an example of  a  substitution that  admits in the spectrum of its matrix   the eigenvalue $1$ and that is  balanced on   factors. 
   Indeed, consider the  Sturmian substitution
   $\tau  \colon 3 \mapsto 30, 0 \mapsto 300$.
   The subshift $(X_{\sigma},T)$ is  deduced from the  Sturmian shift  $(X_{\tau},T)$ (which is balanced on factors by \cite{FagnotVuillon:02}) by applying  the  substitution     $\varphi \colon 0 \mapsto 21, 3 \mapsto 3$, which preserves balancedness.
    \end{example}

\subsection{Coboundaries and  topological eigenvalues } \label{subsec:coboundaries}
We  now recall   a convenient  topological interpretation of the notion of  balancedness in terms of   coboundaries and   topological eigenvalues.

Let $(X,T)$  be a topological dynamical system, that is, $T$ is a homeomorphism acting on the compact   space $X$. Consider e.g. $(X_{\sigma},T)$ for some primitive substitution  $\sigma$ or some subshift $(X,T)$.  Recall that it is said to be \emph{minimal}  if  every non-empty   closed  $T$-invariant subset   of $X$ is  equal to $X$.  We  let denote  by $C(X,\mathbb{R})$ the additive group of continuous maps from $X$ to $\mathbb{R}$ and by  $C(X,\mathbb{Z})$ the additive group of continuous maps from $X$ to $\mathbb{Z}$. We take $\Z$ as a topological space with the discrete topology, so that any $f\in C(X,\Z)$ is {\it locally constant}: for all $x\in X$, there is an open neighborhood $U$ of $x$ such that $y\in U$ implies $f(y)=f(x)$. Since $X$ is a compact space, any $f\in C(X,\Z)$ takes a finite number of values. These two conditions imply that, for any $f\in C(X,\Z)$, there exists $R>0$ such that for all $x,y\in X$, $d(x,y)\leq R $ implies $  f(x)=f(y)$. If $(X,T)$ is a subshift, then, for all $f\in C(X,\Z)$, there exists a positive integer $k$ such that   $x_{[-k,k]}= y_{[-k,k]}$
 implies $  f(x)=f(y)$.

For every $g\in C(X,\mathbb{R})$, we define the {\it coboundary} of $g$ by
\begin{equation}\label{coboundary}
\partial g= g\circ T - g. 
\end{equation} 
The map $g\mapsto \partial g$ is an endomorphism of $C(X,\mathbb{R})$. When an element $f$ belongs to $\partial C(X,\mathbb{R})$, we say that {\it f is a coboundary}. Two functions $f,g$ are said to be {\it cohomologous} if $f-g$ is a coboundary.
For any  non-negative integer $n$,  $f^{(n)}$ stands for  the  map    in $C(X,\mathbb{R})$  defined    for any $x \in X $ as    $$f^{(n)} (x): =f(x)+f\circ T(x)+\cdots+f\circ T^j(x)+\cdots+ f\circ T^{n-1}(x).$$
The family   of maps  $(f^{(n)})_{n \in {\mathbb N}}$  is  called the {\em  cocycle}    of $f$.
The following    theorem states that   being a coboundary means  having    a bounded cocycle.

\begin{theorem}[Gotshalk-Hedlund's Theorem \cite{GotHed:55}]\label{GH} Let $(X,T)$  be a minimal topological dynamical system. The map  $f\in C(X,\mathbb{R})$ is a coboundary if and only in there exists $x_0\in X$ such that the sequence $(f^{(n)}(x_0))_{n\in\mathbb{N}}$ is bounded.
\end{theorem}

Let $(X,T)$ be a 
topological dynamical system, where $T$ is a homeomorphism.
A non-zero complex-valued continuous function $f\in C(X,\mathbb{C})$
is an {\em eigenfunction} for $(X,T)$ if there exists
$\lambda \in  {\mathbb C} $ such that
$\forall x \in X, \ f(Tx)=\lambda f(x).$
The  eigenvalues corresponding to those eigenfunctions
are called the {\em continuous
eigenvalues} of $(X,T)$. If $\theta $ is such that $e^{2i \pi \theta}$ is an   eigenvalue, 
$\theta$ is said to be an {\em additive topological   eigenvalue}. 

As a direct consequence of  Theorem \ref{GH},  we now can reformulate     balancedness in spectral terms.

\begin{proposition}\label{eq_cob}
Let  $(X,T)$  be a  minimal and uniquely ergodic  subshift and let $\mu$ stand for its invariant measure. Given a factor $v\in\mathcal{L}_X$, define
\begin{equation}\label{fv}
f_v = \chi_{[v]}-\mu([v])\in C(X,\mathbb{R}),
\end{equation}
where $\chi_{[v]}$ stands for the characteristic function of the cylinder $[v]$.
Then, $(X,T)$ is balanced on the factor $v$ if and only if the map  $f_v$ is a coboundary. 
 If $\sigma$ is  balanced on the factor $v$, then $\mu[v]$ is  an additive  topological  eigenvalue of $(X,T)$.
\end{proposition}

\begin{proof}
We assume that  $X$ is balanced on the factor $v$.  By  Theorem \ref{GH},  there exists $g \in  C(X,\mathbb{R})$  such that 
$f_v = g \circ T-g $.
Note that $e^{2i\pi  \chi_{[v]} (u)}=1$ for  any $u\in X$.
This yields
$$\exp^{2i\pi g\circ T}= \exp ^{2i\pi  \mu_v } \exp  ^{ 2i \pi  g}.$$
Hence  $\exp  ^{ 2i \pi  g}$ is a topological   eigenfunction  associated with the  additive topological  eigenvalue $ \mu_v $.
\end{proof}

We will also  need the following statement in  Section \ref{sec:sub}.

\begin{proposition}\label{entiers}
Let $(X,T)$ be a minimal topological dynamical system. If $f\in C(X,\mathbb{Z})$ is a coboundary of some  function $g\in C(X,\mathbb{R})$, then it is the coboundary of some function  $h\in C(X,\mathbb{Z})$.
\end{proposition}
\begin{proof} We recall the proof of \cite{Host:1995,DHP}.  Let $\mathbb{T}=\mathbb{R}/\mathbb{Z}$ be the one-dimensional torus and $\pi:\mathbb{R}\to \mathbb{T}$ the canonical projection. Let $\tilde{\partial}$ denote the coboundary map defined on $C(X,\mathbb{T})$ in the same way as in \eqref{coboundary}. Note first that if $\gamma\in C(X,\mathbb{T})$ and $\tilde{\partial}\gamma=0$, then $\gamma$ is constant. Indeed, let $\tilde{c}\in\mathbb{T}$ and set $Y=\gamma^{-1}(\{\tilde{c}\})$. The subset $Y$ is closed since $\gamma$ is continuous and it is $T-$invariant since $\tilde{\partial}\gamma=0$. The system being minimal, if $Y$ is nonempty, it is necessarily the whole space $X$.\\
Suppose $f\in C(X,\mathbb{Z})$ is the coboundary of some $g\in C(X,\mathbb{R})$. Then, $g\circ T(x)-g(x)\in \mathbb{Z}$ for all $x\in X$. This implies that $\tilde{\partial}(\pi\circ g)=0$ and then there exists $\tilde{c}\in\mathbb{T}$ such that $\pi\circ g(x)=\tilde{c}$  for all $x\in X$. Let $c$ be any element in $\pi^{-1}(\{\tilde{c}\})$ and define $h(x):=g(x)-c$. Since $\pi\circ h=0$, $h\in C(X,\mathbb{Z})$, and it is clear that $\partial h=\partial g= f$.  
\end{proof}

\section{Balancedness   of  dendric words} \label{sec:dendric}

Dendric subshifts are minimal   subshifts defined with respect to combinatorial properties of their language expressed  in terms of  extension graphs, such as  recalled in Section \ref{subsec:dendric}.  Elements  of dendric subshifts   are also called  tree words (see e.g.~\cite{BDFDLPRR:15,BDFDLPRR:15bis,BFFLPR:2015,BDFDLPR:16,Rigidity:17}).  We use the terminology  dendric subshifts   in order to avoid  any ambiguity with respect to the notion of   tree  shift that   refers to 
shifts defined on trees  (see e.g. \cite{AubrunBeal}). We   consider balancedness  for dendric subshifts in Section \ref{subsec:baldend} and prove 
Theorem~\ref{theo:equilibre}.
This class of subshifts   encompasses  subshifts generated by interval exchanges,  as well as Sturmian and Arnoux--Rauzy subshifts  discussed in Section  \ref{subsec:exdendric}.  They  have linear factor complexity.

\subsection{Dendric subshifts} \label{subsec:dendric} 
Extension graphs  are bipartite graphs that   describe  the left and right extensions of factors and dendric subshifts are such that all their extension graphs are trees. 
More precisely, let $(X,T)$ be a subshift on the alphabet ${\mathcal A}$.
For $w \in {\mathcal L} (X) $, we  let denote as 
$$
 \begin{array}{lcl}
 L(w) = \{ a \in {\mathcal A} \mid aw \in  {\mathcal L} (X) \},		\\
R(w) = \{ a \in {\mathcal A} \mid wa \in {\mathcal L} (X)\},	\\
 E(w) = \{ (a,b) \in {\mathcal A} \times {\mathcal A} \mid awb \in {\mathcal L} (X)\}.
 \end{array}
$$

For a word $w \in F$, we consider the undirected bipartite graph ${\mathcal E}(w)$, called its \emph{extension graph},  defined as follows:
its set of vertices is the disjoint union of $L(w)$ and $R(w)$ and its edges are the pairs $(a,b) \in E(w)$.
For an illustration, see Example~\ref{ex:fibo} below.
A  minimal subshift $(X,T)$ is a \emph{dendric subshift} if, for every word $w \in {\mathcal L}(X)$, the graph $\E(w)$ is a tree.

\begin{example}
\label{ex:fibo} Let $\sigma_F$ be the Fibonacci substitution defined  over $\{0,1\}$ by  $\sigma_F \colon 0 \mapsto 01, 1 \mapsto 0$. 
The extension graphs of the empty word and of the  letters  $0$ and $1$  are depicted in Figure~\ref{fig:fibo-ext}.

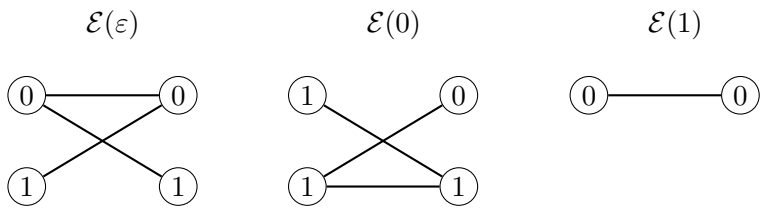
\begin{figure}[h]
 \tikzset{node/.style={circle,draw,minimum size=0.5cm,inner sep=0pt}}
 \tikzset{title/.style={minimum size=0.5cm,inner sep=0pt}}

 \begin{center}
  \begin{tikzpicture}
   \node[title](ee) {$\E(\varepsilon)$};
   \node[node](eal) [below left= 0.5cm and 0.6cm of ee] {$0$};
   \node[node](ebl) [below= 0.7cm of eal] {$1$};
   \node[node](ear) [right= 1.5cm of eal] {$0$};
   \node[node](ebr) [below= 0.7cm of ear] {$1$};
   \path[draw,thick]
    (eal) edge node {} (ear)
    (eal) edge node {} (ebr)
    (ebl) edge node {} (ear);
   \node[title](ea) [right = 3cm of ee] {$\E(0)$};
   \node[node](aal) [below left= 0.5cm and 0.6cm of ea] {$1$};
   \node[node](abl) [below= 0.7cm of aal] {$1$};
   \node[node](aar) [right= 1.5cm of aal] {$0$};
   \node[node](abr) [below= 0.7cm of aar] {$1$};
   \path[draw,thick]
    (aal) edge node {} (abr)
    (abl) edge node {} (aar)
    (abl) edge node {} (abr);
   \node[title](eb) [right = 3cm of ea] {$\E(1)$};
   \node[node](bal) [below left= 0.5cm and 0.6cm of eb] {$0$};
   \node[node](bar) [right= 1.5cm of bal] {$0$};
   \path[draw,thick]
    (bal) edge node {} (bar);
  \end{tikzpicture}
 \end{center}

 \caption{The extension graphs of $\varepsilon$ (on the left), $0$ (on the center) and $1$ (on the right) are trees.}
 \label{fig:fibo-ext}
\end{figure}
\end{example}

\subsection{Balancedness for  dendric subshifts} \label{subsec:baldend}
The main result    of this section  is  Theorem \ref{theo:equilibre}, whose proof  relies on  Lemmas \ref{biptree} and \ref{lem:H} stated and proved  below. 
\begin{lemma}\label{biptree}
Let $\T$ be a finite tree, with a bipartition  $X$ and $Y$ of  its  set of vertices,  with $\Card(X),\Card(Y)\geq 2$. Let $E$ stand for   its set of edges. For all $x\in X$, $y\in Y$, define
$$Y_x:=\{y\in Y: (x,y)\in E\} \qquad X_y:=\{x\in X: (x,y)\in E\}.$$
Let $(G,+)$ be an abelian group and $H$ a subgroup of $G$. Suppose  that there exists a function $g:X\cup Y\cup E\to G$ satisfying the following conditions:
\begin{itemize}
\item[(1)]$g(X\cup Y)\subseteq H$;
\item[(2)] for all $x\in X$, $g(x)=\sum_{y\in Y_x}g(x,y)$, and for all $y\in Y$, $g(y)=\sum_{x\in X_y}g(x,y)$.
\end{itemize}
Then, for all $(x,y)\in E$, $g(x,y)\in H$.
\end{lemma}

\begin{proof} Observe  first that Conditions (1) and (2) imply that the image under $g$ of any edge connected to a leaf belongs to $H$.

We proceed by induction on $k:=\max\{\Card(X),\Card(Y)\}$
and we first assume  $k=2$.  Such as illustrated in Figure~\ref{fig:biptree}, there is only one possibility for the graph  $\T$ (modulo a relabeling of  the vertices), since $\T$ is connected and has no cycles, which is
$$X=\{x_1,x_2\}, Y=\{y_1,y_2\}, E=\{(x_1,y_1),(x_2,y_1),(x_2,y_2)\}.$$

\begin{figure}[h]
 \tikzset{node/.style={circle,draw,minimum size=0.7cm,inner sep=0pt}}
 \tikzset{title/.style={minimum size=0.1cm,inner sep=0pt}}

 \begin{center}
  \begin{tikzpicture}
   \node[title](ee) {};
   
   \node[node](eal) [below left= 0.5cm and 0.6cm of ee] {$x_1$};
   \node[node](ebl) [below= 0.7cm of eal] {$x_2$};
   \node[node](ear) [right= 1.5cm of eal] {$y_1$};
   \node[node](ebr) [below= 0.7cm of ear] {$y_2$};
   \path[draw,thick]
    (eal) edge node {} (ear)
    (ebl) edge node {} (ear)
	(ebl) edge node {} (ebr);

	 	 \node[title](eal) [left = 0.3cm of eal] {\bf \textleaf};
	 \node[title](ebr) [right = 0.3cm of ebr] {\bf \textleaf};
  \end{tikzpicture}
 \end{center}

 \caption{The tree $\mathcal{T}$ when $k=2$.}
 \label{fig:biptree}
\end{figure}
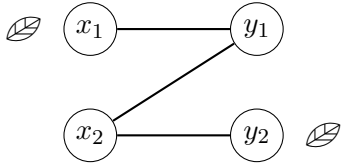

Both $g(x_1,y_1)$ and $g(x_2,y_2)$ are in $H$ because $x_1$ and $y_2$ are leaves. By Condition (2), one has
$g(x_2)=g(x_2,y_1)+g(x_2,y_2)$, and then $g(x_2,y_1)=g(x_2)-g(x_2,y_2)$. 
Since $g(x_2)\in H$ by Condition (1) and $H$ is a group, then  $g(x_2,y_1)\in H$.

We now assume  $k>2$ and that the induction hypothesis  holds   for $k-1$. Suppose also {\em wlog}  that $\Card(X)\geq \Card(Y).$
Note that in this case there exists a leaf in $X$. Indeed, if all vertices in $X$ have degree at least 2, then
$$\Card(E)=\sum_{x\in X}\deg(x)\geq 2\Card(X)$$
because $\T$ is a bipartite graph. On the other hand, since $\T$ is a tree,
$$\Card(E)=\Card(X)+\Card(Y)-1<\Card(X)+\Card(Y)\leq 2\Card(X)$$
which yields the desired  contradiction. The same argument shows that if $X$ and $Y$ have the same cardinality, then both  $X$  and $Y$  have at least one leaf.
We  distinguish two cases, namely  $\Card(X)>\Card(Y)$ and $\Card(X)=\Card(Y)$.

\bigskip

 We first  assume   that $\Card(X)>\Card(Y)$.
Take a leaf in $X$, and  call it $x_0$. Consider the graph $\widetilde{\T}$ obtained from $\T$ by removing the vertex $x_0$ and the edge $(x_0,y_0)$, where $y_0$ is the only vertex in $Y$ connected with $x_0$. This new graph is also a tree, with bipartition of vertices $\widetilde{X}=X-\{x_0\}$, $\widetilde{Y}=Y$, and set of edges $\widetilde{E}=E-\{(x_0,y_0)\}$. Since $\Card(\widetilde{X})=k-1$ and $\Card(\widetilde{Y})=\Card(Y)$, then $\max\{\Card(\widetilde{X}),\Card(\widetilde{Y})\}=k-1$. 

We define $\widetilde{g}$ in  $\widetilde{X}\cup \widetilde{Y}\cup \widetilde{E}$ as follows. On $(\widetilde{X}\cup \widetilde{Y}\cup \widetilde{E})-\{y_0\}$, $\widetilde{g}=g$; on $y_0$, define $\widetilde{g}(y_0)=g(y_0)-g(x_0,y_0)$.
 Let us verify that $\widetilde{g}$ satisfies Conditions (1) and (2) with respect to  $\widetilde{\T}$.

 \begin{itemize}
  
\item[(1)] If $x\in\widetilde{X}$, $\widetilde{g}(x)=g(x)\in H$. If $y\in\widetilde{Y}$ and $y\neq y_0$, $\widetilde{g}(y)=g(y)\in H$. If $y=y_0$, then $\widetilde{g}(y_0)=g(y_0)-g(x_0,y_0)$, but both $g(y_0)$ and $g(x_0,y_0)$ are in $H$, since $g$ satisfies Conditions (1) and (2), and $x_0$ is a leaf. Therefore, the image under $\widetilde{g}$ of any vertex of $\widetilde{\T}$ is in $H$.

\item[(2)] We need a more precise notation here. For a vertex $x\in \widetilde{X}$, we define
$Y_x^{\T}:=\{y\in Y: (x,y)\in E\}$ and $ Y_x^{\widetilde{\T}}:=\{y\in \widetilde{Y}: (x,y)\in \widetilde{E}\}.$
If $x\in \widetilde{X}$, then $Y_x^{\T}=Y_x^{\widetilde{\T}}$, and for all $y\in Y_x^{\T}$, $\widetilde{g}(x,y)=g(x,y)$. Therefore,
$$\widetilde{g}(x)=g(x)=\sum_{y\in Y_x^{\T}} g(x,y)=\sum_{y\in Y_x^{\widetilde{\T}}} \widetilde{g}(x,y).$$
We use  analogously  the  notation $X_y^{\T}$ and $X_y^{\widetilde{\T}}$ for a vertex $y\in \widetilde{Y}$. Let be $y\in \widetilde{Y}$. If $y\neq y_0$, then $X_y^{\T}=X_y^{\widetilde{\T}}$ and for all $x\in X_y^T$, $\widetilde{g}(x,y)=g(x,y)$. Hence,
$$\widetilde{g}(y)=g(y)=\sum_{x\in X_y^{\T}} g(x,y)=\sum_{x\in X_y^{\widetilde{\T}}} \widetilde{g}(x,y).$$
Finally, if $y\in \widetilde{Y}$ and $y=y_0$, then $X_y^{\T}=X_y^{\widetilde{\T}}\cup \{x_0\}$.  We thus have 
$$\begin{array}{ll}
\widetilde{g}(y)=g(y_0)-g(x_0,y_0)&=-g(x_0,y_0)+\sum_{x\in X_y^{\T}} g(x,y)\\&=-g(x_0,y_0)+g(x_0,y_0)+\sum_{x\in X_y^{\widetilde{\T}}} \widetilde{g}(x,y)
= \sum_{x\in X_y^{\widetilde{\T}}} \widetilde{g}(x,y),
\end{array}$$
which ends the proof of the fact  that $\widetilde{g}$ satisfies Conditions (1) and (2).
\end{itemize}

Hence, by  induction,  for all $(x,y)\in \widetilde{E}$, $\widetilde{g}(x,y)\in H$. But in $\widetilde{E}$ one has $\widetilde{g}=g$, which implies that for all $(x,y)\in \widetilde{E}$, $g(x,y)\in H$. Since $x_0$ is a leaf in $X$, $g(x_0,y_0)\in H$, and then for all $(x,y)\in E$, $g(x,y)\in H$.   This ends the   case $\Card(X)>\Card(Y)$.

\bigskip

 We now assume  that $\Card(X)=\Card(Y)$.   Then, both $X$ and $Y$ have at least one leaf; let us call them $x_0$ and $y_0$, respectively. Let $x_{y_0}$ and $y_{x_0}$ denote the only vertices connected with $x_0$ and $y_0$, respectively. It is not difficult to see that $y_0\neq y_{x_0}$ and $x_0\neq x_{y_0}$, since $\T$ is connected and has no cycles.

Consider the graph $\widetilde{\T}$ obtained from $\T$ by removing the vertices $x_0$ and $y_0$, and the edges $(x_0,y_{x_0})$ and $(x_{y_0},y_0)$. This new graph is again a  tree, with bipartition of vertices $\widetilde{X}=X-\{x_0\}$, $\widetilde{Y}=Y-\{y_0\}$, and set of edges $\widetilde{E}=E-\{(x_0,y_{x_0}),(x_{y_0},x_0)\}$. Since $\Card(\widetilde{X})=k-1$ and $\Card(\widetilde{Y})=k-1$, then $\max\{\Card(\widetilde{X}),\Card(\widetilde{Y})\}=k-1$.

 On the new set $\widetilde{X}\cup \widetilde{Y}\cup \widetilde{E}$, define the function $\widetilde{g}$ as follows. On $(\widetilde{X}\cup \widetilde{Y}\cup \widetilde{E})-\{x_{y_0},y_{x_0}\}$, $\widetilde{g}=g$; on $x_{y_0}$, define $\widetilde{g}(x_{y_0})=g(x_{y_0})-g(x_{y_0},y_0)$, and on $y_{x_0}$, $\widetilde{g}(y_{x_0})=g(y_{x_0})-g(x_0,y_{x_0})$.

 Following the same strategy as in the case $\Card(X)>\Card(Y)$, one can see that $\widetilde{g}$ satisfies Conditions (1) and (2) in $\widetilde{\T}$, and since $\max\{\Card(\widetilde{X}),\Card(\widetilde{Y})\}=k-1$, we conclude by induction
 that for any edge $(x,y)\in \widetilde{E}$, $\widetilde{g}(x,y)$ belongs to $H$, which implies that $g(x,y)\in H$. Since $x_0$ and $y_0$ are leaves in $X$ and $Y$, $g(x_0,y_{x_0}), g(x_{y_0},y_0)\in H$. We conclude that for all $(x,y)\in E$, $g(x,y)\in H$.    
\end{proof}

\begin{lemma}\label{lem:H}
Let $(X,T)$ be a minimal dendric subshift. 
 Let $H$ be the following subset of $C(X,\mathbb{Z})$:
$$H=\left\lbrace\sum_{a\in\mathcal{A}}\sum_{k\in K_a}\alpha(a,k)\chi_{T^k([a])}: K_a\subseteq \mathbb{Z},|K_a|<\infty, \alpha(a,k)\in\mathbb{Z}\right\rbrace,$$
where $\chi_{A}$ denotes the characteristic function of the set  $A$, for all $A\subseteq X$. 
Then, for all $v\in\mathcal{L}(X)$, the characteristic function $\chi_{[v]}$ belongs to $H$.
\end{lemma}

\begin{proof} One first easily checks that $H$ is a subgroup.  We now  proceed by induction on the length of $v$. 
The claim  is clearly  true if $|v|=1$, that is, when $v$ is a letter of $\mathcal{A}$, by setting  $K_a=\{0\}$ and $\alpha(a,k)=1$ if $a=v$, $0$ otherwise.
Now suppose that for all $u\in \mathcal{L}(X)$ with $|u|\leq n$, one has $\chi_{[u]}\in H$. Let $v$ be a word  of length  $n+1$.  We write $$v=v_0\cdots v_n \mbox{  and define } \widetilde{v}=v_1\cdots v_n, \  v'=v_0\cdots v_{n-1}, \ \bar{v}=v_1\cdots v_{n-1}.$$ We analyze separately three cases depending on the right/left extensions of $\widetilde{v}$ and $v'$, namely  $l(\widetilde{v})=1$, $r(v')=1$, and  then, finally, $l(\widetilde{v})\geq 2$ and $r(v')\geq 2$, with this latter case  being handled thanks to Lemma  \ref{biptree}.
\begin{itemize}
\item Suppose first  that  $l(\widetilde{v})=1$.
The only left extension of $\widetilde{v}$ is $v_0$, and thus, for all $x\in X$, $\chi_{[v]}(x)=\chi_{[\widetilde{v}]}(Tx)$. By induction hypothesis we have that $\chi_{[\widetilde{v}]}$ belongs to $H$, so we obtain that for all $x\in X$,
$$\chi_{[v]}(x)=\sum_{a\in\mathcal{A}}\sum_{k\in K_a}\alpha(a,k)\chi_{T^{k-1}([a])}(x).$$
Defining $K_a':=\{k-1:k\in K_a\}$ for all $a\in \mathcal{A}$, and $\beta(a,k)=\alpha(a,k+1)$ for all $k\in K_a'$, we conclude that, for all $x\in X$,
$$\chi_{[v]}(x)=\sum_{a\in\mathcal{A}}\sum_{k\in K_a'}\beta(a,k)\chi_{T^{k}([a])},$$
and then $\chi_{[v]}$ belongs to $H$.

\item Now suppose that $r(v')=1$.  The only right extension of $v'$ is $v_n$,  and thus,  for all $x\in X$, $\chi_{[v]}(x)=\chi_{[v']}(x)$. We conclude by applying the induction  hypothesis.

\item  Finally, we assume $l(\widetilde{v})\geq 2$ and $r(v')\geq 2$. Let   $\mathcal{E}(\bar{v})$ be the extension graph of $\bar{v}$ (as defined in Section \ref{subsec:dendric}). It is a tree by definition,
 and each of the sets in its bipartition of vertices has cardinality at least two.

 Define $g:L(\bar{v})\cup R(\bar{v})\cup E(\bar{v})\to G$ as follows.
For $a\in L(\bar{v})$, $g(a)=\chi_{[a\bar{v}]}$, for $b\in R(\bar{v})$, $g(b)=\chi_{T^{-1}[\bar{v}b]}$, and for $(a,b)\in E(\bar{v})$, $g(a,b)=\chi_{[a\bar{v}b]}$. 
Condition (1)  of  Lemma \ref{biptree} holds by induction hypothesis.  

Let us check  that (2) holds. 
Let $a\in L(\bar{v})$. One has 
$$\chi_{[a\bar{v}]}=\sum_{b\in R(\bar{v}),(a,b)\in E(\bar{v})}\chi_{[a\bar{v}b]}(x) \quad  \mbox{ and thus }  \quad
g(a)=\sum_{b\in R(\bar{v}),(a,b)\in E(\bar{v})} g(a,b).$$
Similarly, let $b\in R(\bar{v})$ and $x\in X$.
One has $$\chi_{T^{-1}[\bar{v}b]}(x)=\chi_{[\bar{v}b]}(Tx)=\sum_{a\in L(\bar{v}),(a,b)\in E(\bar{v})}\chi_{[a\bar{v}b]}(x).$$
We conclude that for all $b\in R(\bar{v})$, $g(b)=\sum_{a\in L(\bar{v}),(a,b)\in E(\bar{v})} g(a,b).$
We now can  apply Lemma  \ref{biptree} which yields that $\chi_{[a\bar{v}b]}\in H$, for any biextension $a\bar{v}b$ of $\bar{v}$. In particular, since $(v_0,v_n)\in E(\bar{v})$,  then $\chi_{[v]}\in H$. 
\end{itemize} 
\end{proof}

\begin{proof} [Proof of Theorem \ref{theo:equilibre}]
We  assume  that  the  dendric subshift  $(X,T)$ is balanced on the letters. Let $ v \in {\mathcal L}(X)$. Let $C$ be a constant of balancedness for the letters. Let $n$ be a positive integer and let $u, w$ be two factors of $\mathcal{L}_X$ of length $n-1$ with  $n-1>|v|$. 
Pick  a  bi-infinite  word $x\in X$ such that $u=x_{[i,i+n)}$ and $w=x_{[j,j+n)}$ for some indices $i,j\in \mathbb{Z}$. We have 
$$||u|_v-|w|_v|=\left|\sum_{l=i}^{i+n-1-|v|}\chi_{[v]}(T^lx)-\sum_{l=j}^{j+n-1-|v|}\chi_{[v]}(T^ly)\right|.$$
Now,   according to Lemma \ref{lem:H}, for all $a\in\mathcal{A}$, let $K_a$ be a finite subset of $\mathbb{Z}$ such that, for all $k\in K_a$, there exists $\alpha(a,k)\in \mathbb{Z}$ verifying
$$\chi_{[v]}=\sum_{a\in\mathcal{A}}\sum_{k\in K_a}\alpha(a,k)\chi_{T^k([a])}.$$
Then,
\begin{eqnarray*}
||u|_v-|w|_v| &=&\left|\sum_{l=i}^{i+n-1-|v|}\sum_{a\in\mathcal{A}}\sum_{k\in K_a}\alpha(a,k)\chi_{T^k[a]}(T^lx)-\sum_{l=j}^{j+n-1-|v|}\sum_{a\in\mathcal{A}}\sum_{k\in K_a}\alpha(a,k)\chi_{T^k[a]}(T^lx)\right|\\
&=& \left|\sum_{a\in\mathcal{A}}\sum_{k\in K_a}\alpha(a,k)\left(\sum_{l=i}^{i+n-1-|v|}\chi_{T^k[a]}(T^lx)-\sum_{l=j}^{j+n-1-|v|}\chi_{T^k[a]}(T^lx)\right)\right|\\
&\leq &  \sum_{a\in\mathcal{A}}\sum_{k\in K_a}|\alpha(a,k)|\left | \sum_{l=i}^{i+n-1-|v|}\chi_{T^k[a]}(T^lx)-\sum_{l=j}^{j+n-1-|v|}\chi_{T^k[a]}(T^lx)\right |\\
&=& \sum_{a\in\mathcal{A}}\sum_{k\in K_a}|\alpha(a,k)|\left | \sum_{l=i}^{i+n-1-|v|}\chi_{[a]}(T^l(T^{-k}x))-\sum_{l=j}^{j+n-1-|v|}\chi_{[a]}(T^l(T^{-k}x))\right |\\
&=&  \sum_{a\in\mathcal{A}}\sum_{k\in K_a}|\alpha(a,k)|\cdot ||(T^{-k}x)_{[i,i+n-|v|)}|_a - |(T^{-k}x)_{[j,j+n-|v|)}|_a|.
\end{eqnarray*}
Note that $(T^{-k}x)_{[i,i+n-|v|)}$ and $(T^{-k}y)_{[j,j+n-|v|)}$ are two factors of length $n-1-|v|$ belonging to $\mathcal{L}(X)$, and then by balancedness on the letters, for all $a\in\mathcal{A}$
$$||(T^{-k}x)_{[i,i+n-|v|)}|_a - |(T^{-k}y)_{[j,j+n-|v|)}|_a|\leq C.$$
We obtain that
$||u|_v-|w|_v|\leq |\mathcal{A}|KC,$
where 
$K=\max_{a\in \mathcal{A}}{\left\lbrace\sum_{k\in K_a}|\alpha(a,k)|\right\rbrace}, $ which ends the proof of the balancedness on $v$.

Lastly, the result on additive   topological eigenvalues comes from Theorem \ref{GH}.
\end{proof}

\subsection{Examples}\label{subsec:exdendric}

Sturmian words, Arnoux-Rauzy words (introduced in \cite{Arnoux-Rauzy:91} and also called episturmian words), and  codings of regular interval exchanges are typical  examples of dendric words (see~\cite{BDFDLPR:16}). See also as an interesting  family of  dendric words, the   words  produced by the Cassaigne--Selmer  multidimensional continued fraction algorithm 
\cite{CasLabLer:17}. Note that dendric words   have  factor complexity  $(d-1)n + 1$ when defined over an alphabet of cardinality $d$  (see~\cite{BDFDLPR:16}). We recall that  the factor complexity of  a  bi-infinite word $u$  counts the number of factors of  $u$ a  given length.

Recall that Sturmian words  are known to be $1$-balanced on letters  \cite{Lothaire:2002}.   They are also known to be   balanced on  factors \cite{FagnotVuillon:02}.  Note that we also  recover this property as a direct consequence of Theorem \ref{theo:equilibre}. 
It was believed that Arnoux-Rauzy words would be 2-balanced on letters, as generalizations of Sturmian words.  But there exist Arnoux-Rauzy words  that are not balanced  on letters, such as  proved in~\cite{Cassaigne-Ferenczi-Zamboni:00}, also see \cite{Cassaigne-Ferenczi-Messaoudi:08}.

More precisely, Arnoux-Rauzy words   are  uniformly recurrent dendric words  that can be expressed as $S$-adic words as follows. Let $\mathcal{A} = \{1,2,\ldots,d\}$.
We  define the set~$\mathcal{S}_{AR}$  of substitutions defined  as $\mathcal{S}_{AR} = \{\sigma_i:\, i \in \mathcal{A}\}$, with
$
\sigma_i:\ i \mapsto i,\ j \mapsto ji\ \mbox{for}\ j \in \mathcal{A} \setminus \{i\}\,.
$
A bi-infinite word $u \in \mathcal{A}^\mathbb{Z}$ is an {\em Arnoux-Rauzy word}  if  its language  coincides with the  language of a word of the form 
$ \lim_{n\to\infty} \sigma_{i_0} \sigma_{i_1} \cdots \sigma_{i_n} (1), 
$
where the sequence ${\mathbf i}=(i_n)_{n\geq 0} \in \mathcal{A}^\mathbb{N}$ is such that every letter in~$\mathcal{A}$ occurs infinitely often in~${\mathbf i}=(i_n)_{n\ge0}$. 
In this latter case, the infinite word $u$ is uniformly recurrent and we can associate with it a two-sided subshift $(X_{\mathbf i},T)$
which contains all the bi-infinite words  having the same language as $u$. 
Furthermore, such a sequence ${\mathbf i}=(i_n)_{n\geq 0} \in \mathcal{A}^\mathbb{N}$ is uniquely defined for a given~$u$.
For any given Arnoux-Rauzy word, the sequence ${\mathbf i}=(i_n)_{n \geq 0}$ is called the $\mathcal{S}_{AR}$-\emph{directive word} of~$u$. 
All the Arnoux-Rauzy words that belong to the dynamical system $(X_{\mathbf i},T)$ have the same $\mathcal{S}_{AR}$-directive word.
An  {\em Arnoux-Rauzy substitution}   is  a  finite product  of substitutions in ${\mathcal S}_{AR}$. 
The following statement  is deduced from Theorem \ref{theo:equilibre}.

\begin{corollary} 
Let $\sigma$ be a primitive  Arnoux-Rauzy substitution. 
Then, $(X_{\sigma},T)$ is  balanced on factors.  

Let $(X_{\mathbf i},T)$ be an Arnoux-Rauzy  subshift  on a three-letter alphabet with   $\mathcal{S}_{AR}$-directive sequence  ${\mathbf i}= (i_n)_{n \geq 0}$. If   there exists some constant  $h$ such that we do not have $i_n= i_{n+1} = \cdots = i_{n+h}$ for any $n \ge 0$, then $(X_{\mathbf i},T)$ is balanced on factors.

In both cases,  frequencies   of  factors are  additive topological eigenvalues and   cylinders are  bounded remainder sets.

\end{corollary}

\begin{proof}
Arnoux-Rauzy substitutions are known to be Pisot  \cite{Arnoux-Ito:01,AvilaDelecroix:13} and thus to generate Arnoux-Rauzy words that are balanced on letters by Theorem \ref{theo:Adam},
and  consequently  on factors by  Theorem \ref{theo:equilibre}.
The  condition of the second statement  is   proved  in \cite{BerthCS:13} to  imply that  $(X_{\mathbf i},T)$ is  $({2h\!+\!1})$-balanced on  letters. We again  conclude thanks to   Theorem \ref{theo:equilibre}.
\end{proof}

It  is  proved in \cite{BerST:18} that on a three-letter alphabet,  a.e.  Arnoux-Rauzy subshift  is balanced on  letters  and  has  pure discrete spectrum in the measure-theoretic sense, that is, 
it   is measurably conjugate to a translation on the torus~$\mathbb{T}^2$.
Here  almost everywhere (a.e.)  refers to some invariant measure;  as an example of  such a   measure,   consider the   measure of maximal entropy for the suspension
flow of the Rauzy gasket    constructed  in \cite{AvilaHubSkripbis} (see also \cite{AvilaHubSkrip}). 
 As an application of 
Theorem \ref{theo:equilibre}, we deduce that   in case of pure discrete spectrum, all cylinders    provide bounded remainder sets  for the underlying  toral translation
and that  for  a.e.  Arnoux-Rauzy subshift, frequencies of factors  are additive  topological eigenvalues.

\section{Balancedness for  substitutions with rational frequencies} \label{sec:sub}
 Trough this section, $\sigma$ will be a primitive substitution on the alphabet $\mathcal{A}$ and $(X_\sigma,T)$ the minimal, uniquely ergodic subshift generated by $\sigma$. Let $\mu$ denote the unique invariant probability measure on $X_\sigma$.
We  first introduce a suitable partition  in  towers  for substitutions  in Section~\ref{subsec:towers},
 we then  provide  criteria for producing imbalancedness in Section~\ref{subsec:criteria}, and lastly, we discuss  several 
examples in Section~\ref{subsec:examples}.

\subsection{Two-letter   towers}\label{subsec:towers}

Let $(X,T)$ be a  subshift. A {\it partition in towers} of $(X,T)$ is a partition of the space $X$ of the form
$$\mathcal{P}=\{T^jB_i:1\leq i\leq m,0\leq j <h_i\}$$
where the $B_i$'s are clopen sets (i.e., closed and open sets)  and non-empty. The number $m$ is the {\it number of towers} of $\mathcal{P}$. For all $1\leq i\leq m$, the subset $\{T^jB_i:0\leq j<h_i\}$ is called the {\it i-th tower} of $\mathcal{P}$; $h_i$ is its {\it height} and $B_i$ its {\it base}. The elements of the partition are called {\it atoms}.
For an illustration, see Figure \ref{partition}.

We recall  below a  classical  description of the subshift $(X_{\sigma}, T)$ in terms of Kakutani-Rohlin partitions provided by  the substitution $\sigma$ which  will  play a  crucial role  in the following (in particular  for   Proposition  \ref{beta_1}).  We  provide the  proof  of the following folklore  result for the sake of self-containedness.   An illustration of the partition $\mathcal{P}_n$  defined  below is provided in Figure \ref{sigma}. 
\begin{lemma}\label{part}
Let $\sigma$ be a primitive substitution. For all $n\in\mathbb{N}$, define 
\begin{equation}\label{pn}
\mathcal{P}_n=\{T^j\sigma^n([ab]):ab\in \mathcal{L}_2(X), 0\leq j < |\sigma^n(a)|\}.
\end{equation}
The sequence $(\mathcal{P}_n)_{n\in\mathbb{N}}$ is a  nested sequence of partitions in towers of $(X_{\sigma},T)$, i.e.,  for all $n\in\mathbb{N}$, $\mathcal{P}_{n+1}$ is finer than $\mathcal{P}_{n}$ and 
$\bigcup _{n,ab\in \mathcal{L}_2(X_{\sigma})} \sigma^{n+1}([ab]) \subset \bigcup_{n ,ab\in \mathcal{L}_2(X_{\sigma})} \sigma^{n}([ab])$. \end{lemma}
\begin{proof}
First we show that for all $n\geq 1$, $\mathcal{P}_n$ covers $X_{\sigma}$.  We fix  $ n \geq 1$ and  let  $x\in X_{\sigma}$. By definition of $X_{\sigma}$, for all $\ell\geq 1$, there exist $N\geq 1$ and $a\in\mathcal{A}$ such that\footnote{We recall that  the notation $ u \prec v$ stands for $u$ being a factor of $v$.}  $x_{[-\ell,\ell)}\prec \sigma^N(a)$. For  all $\ell$ large, one has  $N>n$, and then $x_{[-\ell,\ell)}\prec \sigma^n(w)$ for some $w\in \mathcal{L}(X_{\sigma})$.
 This implies that there exist $0\leq j<|w|$ and $0\leq k,k'\leq \max_{a\in\mathcal{A}}\{|\sigma^n(a)|\}$ such that $x_{[-\ell+k,\ell-k')}=\sigma^n(w_j)$.  
Since $|w|\to\infty$ as $\ell\to\infty$, a Cantor diagonal argument provides a word $y\in X_{\sigma}$ and an integer $k$ with  $0\leq k< |\sigma^n(y_0)|$ such that $x=T^k\sigma^n(y)$. Setting $ab=y_0y_1$, we have that $ab\in \mathcal{L}_2(X_{\sigma})$ and $x\in T^k\sigma^n([ab])$.

We now prove that ${\mathcal P}_n$ is a partition. Suppose that  there exist $ab, cd\in \mathcal{L}_2(X_{\sigma}),$ with $  0\leq j<|\sigma^n(a)|$ and $0\leq k<|\sigma^n(c)|$, such that $x\in T^j\sigma^n([ab])\cap T^k\sigma^n([cd])$. Then, $x=T^j\sigma^n(y_1)=T^k\sigma^n(y_2)$, where $y_1\in [ab]$, $y_2\in [cd]$. By recognizability (we use the fact that $\sigma$ is primitive), $j=k$ and $ac=bd$, so in fact $T^j\sigma^n([ab])=T^k\sigma^n([cd])$. 

Finally, let us show that $\mathcal{P}_{n+1}$ is finer than $\mathcal{P}_{n}$. Let $T^k\sigma^{n+1}([ab])$ be an atom of $\mathcal{P}_{n+1}$, and  let $x$ belong to it. There exists $y\in [ab]$ such that $x=T^k\sigma^{n+1}(y)$, and therefore $x$ belongs also to $T^k\sigma^{n}([cd])$, where $c=\sigma(y)_0$, $d=\sigma(y)_1$. By definition of $\mathcal{P}_{n+1}$, one has  $0\leq k< |\sigma^{n+1}(a)|$. If $0\leq k<|\sigma^n(\sigma(a)_0)|=|\sigma^n(c)|$, then $T^k\sigma^{n}([cd])$ is an atom of $\mathcal{P}_n$ and we conclude that $T^k\sigma^{n+1}([ab])$ is contained in an atom of $\mathcal{P}_n$. If $|\sigma^n(\sigma(a)_0)|\leq k<|\sigma^{n+1}(a)|$, then there is a unique $j$ with  $1\leq j<|\sigma(a)|$ such that
$$|\sigma^{n}(\sigma(a)_{[0,j)})|\leq k < |\sigma^{n}(\sigma(a)_{[0,j+1)})|.$$
Define $m=|\sigma^{n}(\sigma(y)_{[0,j)})|$. We know that $T^m\sigma^n(\sigma(y))=\sigma^nT^j(\sigma(y))$, so we conclude that
$$x=T^k\sigma^{n}\sigma(y)=T^{k-m}\sigma^nT^j\sigma(y)\in T^{k-m}\sigma^n([\sigma(y)_j\sigma(y)_{j+1}]).$$
Since $0\leq k-m < |\sigma^n(\sigma(y)_j)|$ and $\sigma(y)_j\sigma(y)_{j+1}\in\mathcal{L}_2(X_{\sigma})$, the subset $T^{k-m}\sigma^n([\sigma(y)_j\sigma(y)_{j+1}])$ is an atom of $\mathcal{P}_n$, so we conclude again that $T^k\sigma^{n+1}([ab])$ is contained in an atom of $\mathcal{P}_n$.
\end{proof}
\begin{remark}\label{rem:part}
Note that we could have used a very similar proof to show that  the  sequence of partitions  $(\mathcal{Q}_n)_{n\in\mathbb{N}}$ with
$$\mathcal{Q}_n=\{T^j\sigma^n([a]):a\in \mathcal{A}, 0\leq j < |\sigma^n(a)|\}$$
is a  nested sequence of partitions in towers of $(X_{\sigma},T)$, see e.g. \cite[Proposition 14]{DHS:99}. However, we are not able to ensure that for every factor $v\in \mathcal{L}(X_{\sigma})$, the function $f_v=  \chi_{[v]}-\mu [v] \in C(X_{\sigma},\mathbb{R}),
$ as defined in Lemma \ref{eq_cob}, will be constant in the atoms of $\mathcal{Q}_{n}$ for  all $n$ large. Indeed, for all $n\geq 1$, the last level of any tower of $\mathcal{Q}_{n}$ determines only the first letter of its elements, unless we put some additional condition on $\sigma$, like being {\it proper} (see \cite{DHS:99}) for details).  We will see in Section \ref{subsec:criteria} that strategies to provide imbalancedness criteria relies on the fact that for any factor $v$ we can always find a positive integer $n$ such that $f_v$ is constant in the atoms of $\mathcal{P}_{n}$.   
\end{remark}

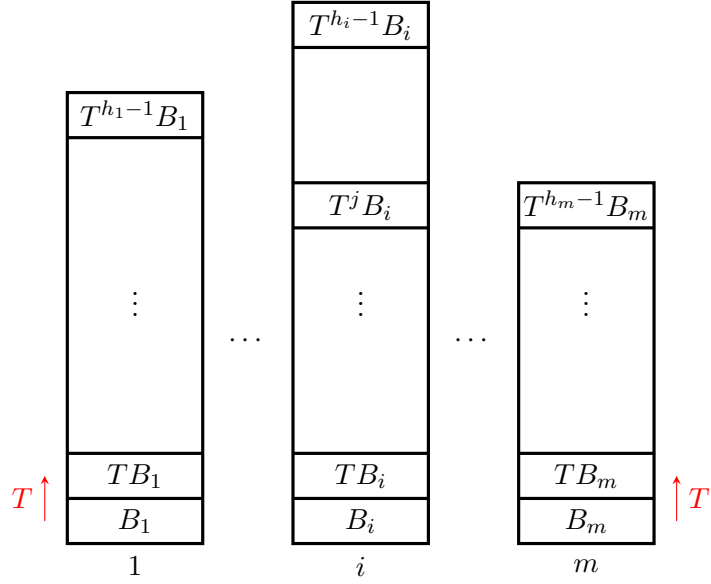
\begin{figure}[h]

 \begin{center}
  \begin{tikzpicture}[scale=0.6]
   
   \draw [very thick] (1,0) rectangle (4,10);
   \draw [very thick] (1,1) -- (4,1);
   \draw [very thick] (1,2) -- (4,2);
   \draw [very thick] (1,9) -- (4,9);
   \node at (2.5,0.5) {$B_1$};
   \node at (2.5,1.5) {$TB_1$};
   \node at (2.5,9.5) {$T^{h_1-1}B_1$};
   
   \draw [very thick] (6,0) rectangle (9,12);
   \draw [very thick] (6,1) -- (9,1);
   \draw [very thick] (6,2) -- (9,2);
   \draw [very thick] (6,7) -- (9,7);
   \draw [very thick] (6,8) -- (9,8);
   \draw [very thick] (6,11) -- (9,11);
   \node at (7.5,0.5) {$B_i$};
   \node at (7.5,1.5) {$TB_i$};
   \node at (7.5,7.5) {$T^{j}B_i$};
   \node at (7.5,11.5) {$T^{h_i-1}B_i$};

   \draw [very thick] (11,0) rectangle (14,8);
   \draw [very thick] (11,1) -- (14,1);
   \draw [very thick] (11,2) -- (14,2);
   \draw [very thick] (11,7) -- (14,7);
   \node at (12.5,0.5) {$B_m$};
   \node at (12.5,1.5) {$TB_m$};
   \node at (12.5,7.5) {$T^{h_m-1}B_m$};

   \node at (2.5,-0.5) {$1$};
   \node at (7.5,-0.5) {$i$};
   \node at (12.5,-0.5) {$m$};    
   
   \node [very thick] at (5,4.5) {$\cdots$};
   \node [very thick] at (10,4.5) {$\cdots$};
   \node at (2.5, 5.5) {$\vdots$};
  \node at (7.5, 5.5) {$\vdots$};
  \node at (12.5, 5.5) {$\vdots$};
  
  \draw [->, >=stealth, red] (0.5,0.5) -- (0.5,1.5); 
  \node [red] at (0,1) {$T$};
  \draw [->, >=stealth, red] (14.5,0.5) -- (14.5,1.5);
  \node [red] at (15,1) {$T$};
  
  \end{tikzpicture}
 \end{center}

 \caption{A partition in towers.}
 \label{partition}
\end{figure}

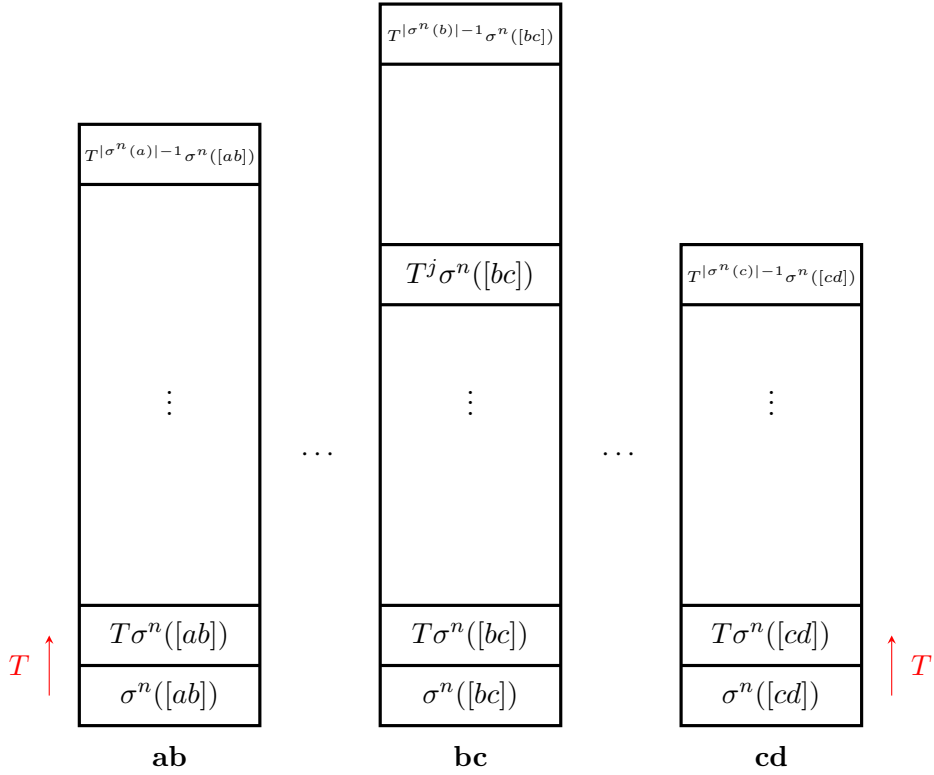
\begin{figure}[h]

 \begin{center}
  \begin{tikzpicture}[scale=0.8]
   
   \draw [very thick] (1,0) rectangle (4,10);
   \draw [very thick] (1,1) -- (4,1);
   \draw [very thick] (1,2) -- (4,2);
   \draw [very thick] (1,9) -- (4,9);
   \node at (2.5,0.5) {$\sigma^n([ab])$};
   \node at (2.5,1.5) {$T\sigma^n([ab])$};
   \node at (2.5,9.5) {\tiny $T^{|\sigma^n(a)|-1}\sigma^n([ab])$};
   
   \draw [very thick] (6,0) rectangle (9,12);
   \draw [very thick] (6,1) -- (9,1);
   \draw [very thick] (6,2) -- (9,2);
   \draw [very thick] (6,7) -- (9,7);
   \draw [very thick] (6,8) -- (9,8);
   \draw [very thick] (6,11) -- (9,11);
   \node at (7.5,0.5) {$\sigma^n([bc])$};
   \node at (7.5,1.5) {$T\sigma^n([bc])$};
   \node at (7.5,7.5) {$T^{j}\sigma^n([bc])$};
   \node at (7.5,11.5) {\tiny $T^{|\sigma^n(b)|-1}\sigma^n([bc])$};

   \draw [very thick] (11,0) rectangle (14,8);
   \draw [very thick] (11,1) -- (14,1);
   \draw [very thick] (11,2) -- (14,2);
   \draw [very thick] (11,7) -- (14,7);
   \node at (12.5,0.5) {$\sigma^n([cd])$};
   \node at (12.5,1.5) {$T\sigma^n([cd])$};
   \node at (12.5,7.5) {\tiny $T^{|\sigma^n(c)|-1}\sigma^n([cd])$};

   \node at (2.5,-0.5) {\bf ab};
   \node at (7.5,-0.5) {\bf bc};
   \node at (12.5,-0.5) {\bf cd};    
   
   \node [very thick] at (5,4.5) {$\cdots$};
   \node [very thick] at (10,4.5) {$\cdots$};
   \node at (2.5, 5.5) {$\vdots$};
  \node at (7.5, 5.5) {$\vdots$};
  \node at (12.5, 5.5) {$\vdots$};
  
  \draw [->, >=stealth, red] (0.5,0.5) -- (0.5,1.5); 
  \node [red] at (0,1) {$T$};
  \draw [->, >=stealth, red] (14.5,0.5) -- (14.5,1.5);
  \node [red] at (15,1) {$T$};

  \end{tikzpicture}
 \end{center}

 \caption{Partition $\mathcal{P}_n$ for a primitive substitution $\sigma$.  The last levels are mapped by $T$ on some  atoms of the base by injectivity. Note that   the elements of a same
 last level are not necessarily mapped to the same atom of the base.}
 \label{sigma}
\end{figure}

\begin{example}[Thue--Morse]\label{ex:TMtower}
We continue Example~\ref{ex:TMsigma2}.
Let $\sigma_{TM}$ be the Thue--Morse substitution on $\{0,1\}$ given by
$\sigma_{TM} \colon 0\mapsto 01$, $1\mapsto 10$.  Let $(\mathcal{P}_n)_{n\geq 1}$ the sequence of partitions in towers defined in \eqref{pn}. Each $\mathcal{P}_n$ has four towers with $2^n$ levels. Two elements in the same atom of $\mathcal{P}_n$ share at least their first $2^n+1$ letters.  The partition  $\mathcal{P}_1$ is depicted in Figure~\ref{fig:partition1TM}.
\end{example}

\begin{figure}[h]

 \begin{center}
  \begin{tikzpicture}[scale=0.6]
      
   \draw [very thick] (1,0) rectangle (4,2);
   \draw [very thick] (1,1) -- (4,1);
   \draw [very thick] (1,2) -- (4,2);
   \node at (2.5,0.5) {$\subseteq [0101]$};
   \node at (2.5,1.5) {$\subseteq  [101]$};
   
  \draw [very thick] (5,0) rectangle (8,2);
   \draw [very thick] (5,1) -- (8,1);
   \draw [very thick] (5,2) -- (8,2);
   \node at (6.5,0.5) {$[0110]$};
   \node at (6.5,1.5) {$ [110]$};

   \draw [very thick] (9,0) rectangle (12,2);
   \draw [very thick] (9,1) -- (12,1);
   \draw [very thick] (9,2) -- (12,2);
   \node at (10.5,0.5) {$[1001]$};
   \node at (10.5,1.5) {$[001]$};
  
   \draw [very thick] (13,0) rectangle (16,2);
   \draw [very thick] (13,1) -- (16,1);
   \draw [very thick] (13,2) -- (16,2);
   \node at (14.5,0.5) {$\subseteq [1010]$};
   \node at (14.5,1.5) {$\subseteq [010]$};    
  
   \node at (2.5,-0.5) {$\bf{00}$};
   \node at (6.5,-0.5) {$\bf 01$};
   \node at (10.5,-0.5) {$\bf 10$};    
   \node at (14.5,-0.5) {$\bf 11$};

  \draw [->, >=stealth, red] (0.5,0.5) -- (0.5,1.5); 
  \node [red] at (0,1) {$T$};
  
  \end{tikzpicture}
 \end{center}

 \caption{Partition $\mathcal{P}_1$ for  the Thue--Morse substitution.}
 \label{fig:partition1TM}
\end{figure}
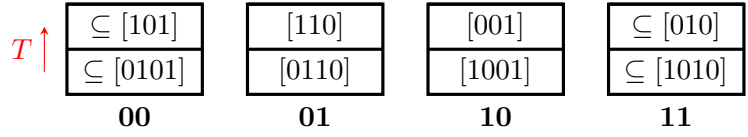

We end this section with  the following lemma which provides a convenient expression for the  entry  $(ab,cd)$  of the two-block   matrix $M_{\sigma_2}$. 

\begin{lemma}\label{m2}
Let $\sigma$   be a primitive substitution. Then, for all $ab, cd\in {\mathcal L}_2(X_{\sigma})$, and for 
all $ n \geq 1$,   \begin{equation}
\tr{M_{\sigma_2}}(ab,cd)=\Card \{0\leq k<|\sigma^{n+1}(a)|: T^k\sigma^{n+1}([ab])\subseteq \sigma^{n}([cd])\}.
\end{equation}
\end{lemma}
\begin{remark}\label{rem:m2}
Lemma \ref{m2} means that the matrix $M_{\sigma_2}$ contains  all  the information for describing  the  transition  from $\mathcal{P}_n$ to $\mathcal{P}_{n+1}$, for all $n$. This corresponds to applying $\sigma$ once. It follows easily by induction that, for all $r\geq 0$, $M_{\sigma_2} ^r$ codes the transition from $\mathcal{P}_n$ to $\mathcal{P}_{n+r}$ and we have
$$(\tr{M_{2}}) ^r (ab,cd)=\Card \{0\leq k<|\sigma^{n+r}(a)|: T^k\sigma^{n+r}([ab])\subseteq \sigma^{n}([cd])\}.
$$
\end{remark}
\begin{proof}
Let $ab, cd\in  {\mathcal L}_2(X_{\sigma})$. We fix $ n \geq 1$.  We recall that $\sigma_2(ab)_j$ stands  for the $j$th letter of $\sigma_2(ab)$ on the alphabet ${\mathcal L}_2(X_{\sigma})$,  with the  first letter  being indexed by $0$.
By definition, one has 
$$\tr{M_{\sigma_2}}(ab,cd)=\Card\{0\leq j<|\sigma(a)|: \sigma_2(ab)_j=cd\}=  \Card\{0\leq j<|\sigma(a)|: T^j\sigma([ab])\subseteq [cd]\}.$$
We thus  want to show that
$$  \Card\{0\leq j<|\sigma(a)|: T^j\sigma([ab])\subseteq [cd]\}=\Card \{0\leq k<|\sigma^{n+1}(a)|: T^k\sigma^{n+1}([ab])\subseteq \sigma^{n}([cd])\}.$$

Suppose  that there exists $0\leq j<|\sigma(a)|$ such that $T^j\sigma([ab])\subseteq [cd]$. If $j=0$, $\sigma([ab])\subseteq [cd]$ and therefore $\sigma^{n+1}([ab])\subseteq \sigma^n([cd])$ ($k=0$).  If $1\leq j<|\sigma(a)|$,    set $k=|\sigma^{n}(\sigma(a)_{[0,j)})|$.  One has $ k<|\sigma^{n+1}(a)|.$ Now take $x\in [ab]$. We have
$T^k\sigma^n(\sigma(x))=\sigma^nT^j(\sigma(x))$.
By hypothesis, $T^j(\sigma(x))\in [cd]$, and then $T^k\sigma^{n+1}(x)\in \sigma^n([cd])$.  
Note also that by definition the $k$ associated with a given $j$ is unique, so we conclude that 
$$  \Card\{0\leq j<|\sigma(a)|: T^j\sigma([ab])\subseteq [cd]\}\geq \Card \{0\leq k<|\sigma^{n+1}(a)|: T^k\sigma^{n+1}([ab])\subseteq \sigma^{n}([cd])\}.$$

Conversely, suppose that  there exists $0\leq k<|\sigma^{n+1}(a)|$ such that $T^k\sigma^{n+1}([ab])$ is included in $\sigma^n([cd])$. Let  $x\in [ab]$  and let  $y=T^k\sigma^{n+1}(x)$. 

We first  assume $0\leq k<|\sigma^{n}(\sigma(a)_0)|$.   
By hypothesis, there exists $z\in [cd]$ such that $y=\sigma^n(z)$. By recognizability, $k=0$ and $\sigma(x)=z$, and thus $\sigma(x)  \in [cd]$. We conclude that $\sigma([ab])\subseteq [cd]$.

Now we assume that  $|\sigma^n(\sigma(a)_0)|\leq k<|\sigma^{n+1}(a)|$. 
 There exists  a unique $j$ with  $1\leq j<|\sigma(a)|$ such that
$$|\sigma^{n}(\sigma(a)_{[0,j)})|\leq k < |\sigma^{n}(\sigma(a)_{[0,j+1)})|.$$
Let  $m=|\sigma^{n}(\sigma(a)_{[0,j)})|$. One has  $T^m\sigma^n(\sigma(x))=\sigma^nT^j(\sigma(x))$,  and thus $y=T^k\sigma^{n}(\sigma(x))=T^{k-m}\sigma^n(T^j\sigma(x)).$
On the other hand, $y=T^k\sigma^{n+1}(x)\in \sigma^n([cd])$, and then, there exists $z\in [cd]$ such that $y=\sigma^n(z)$.  One has $ 0 \leq k-m <  |\sigma^{n}(\sigma(a)_{j})| $.  By recognizability, $k-m=0$ and $T^j\sigma(x)=z\in[cd]$. We conclude that $T^j\sigma([ab])\subseteq [cd]$.
Since the integer $j$ associated with a given $k$ is unique, we conclude that 
$$  \Card\{0\leq j<|\sigma(a)|: T^j\sigma([ab])\subseteq [cd]\}\leq \Card \{0\leq k<|\sigma^{n+1}(a)|: T^k\sigma^{n+1}([ab])\subseteq \sigma^{n}([cd])\}.$$
\end{proof}

\subsection{Some criteria for detecting imbalancedness} \label{subsec:criteria}

   For any $n\geq 1$,   let   $R_n(X_{\sigma})$ (resp. $Z_n(X_\sigma)$)   be  the set of maps from ${\mathcal L}_n(X_\sigma)$ to $\mathbb{R}$ (resp. to $\mathbb{Z}$) and let  $ \beta$   be the map defined as  $$\beta:R_1(X_\sigma)\to R_2(X_\sigma),  \quad 
f \mapsto (\beta f )(ab)=f(b)-f(a) \mbox{   for all  } ab\in {\mathcal L}_2(X_\sigma).$$

Our strategy works as follows.  We consider the map  $f_v= \chi_{[v]}-\mu([v])=\chi_{[v]}-\mu_v $ such as defined
in~(\ref{fv}).  We will  use the fact that  the  map  $f_v$  is constant in the atoms of  the two-letter partition 
 ${\mathcal P}_n$   (defined in  (\ref{pn})) for   all $n$ large,  and associate with $f_v$  a map $\phi_{v,n} \in R_2(X_{\sigma})$, thus defined on    $\mathcal{L}_2(X_{\sigma})$.
 Proposition \ref{beta_1} first  provides a  convenient necessary condition on such a  map, namely, it is proved to  belong to   $\beta(R_1(X_{\sigma}))$.
 This condition 
  is translated  in symbolic  terms in Proposition~\ref{t1}, and then exploited in Theorem \ref{condnec1}. 
 Indeed, knowing that a  map  belongs  to   $ \beta(R_1(X_{\sigma}))$  implies several  convenient restrictions,  for instance   its  coordinate on  each factor of the form 
 $aa$ is equal to $0$.
 The proof of Proposition \ref{beta_1}  below  closely follows the approach developed in   \cite{Host:1995,Host:2000,DHP}. 
 Corollary~\ref{cor:TM}  illustrates how powerful this  simple  formulation can be.

\begin{proposition}\label{beta_1}
Let $\sigma$ be a  primitive substitution. Let $f\in C(X_{\sigma},\mathbb{Z})$ such that there exists $k\in \mathbb{N}$ for which $f$ is constant in the atoms of $\mathcal{P}_k$. For all $n\geq k$, define $\phi_n\in \mathbb{R}^{\mathcal{L}_2(X_{\sigma})}$ by
\begin{equation}\label{phin}
\phi_n(ab)=\sum_{j=0}^{|\sigma^n(a)|-1} f\mid_{T^j\sigma^n([ab])} \quad \forall ab\in \mathcal{L}_2(X_{\sigma}).
\end{equation}
Let $d=|\mathcal{L}_2(X_{\sigma})|$. If $f$ is a coboundary, then $\phi_n\in \beta(R_1(X_{\sigma}))$ for all $n\geq k+d$. 
\end{proposition}
\begin{proof} Let $f\in C(X_{\sigma},\mathbb{Z})$ such that there exists $k\in \mathbb{N}$ for which $f$ is constant in the atoms of $\mathcal{P}_k$. Suppose that  $f$ is a coboundary, that is, there exists $g\in C(X_{\sigma},\mathbb{R})$ such that $f=g\circ T-g$. By Proposition \ref{entiers}, $g\in C(X_{\sigma},\mathbb{Z})$, and then it is locally constant. We claim that there exists $\ell \geq k+d$ such that for all $ab\in \mathcal{L}_2(X_{\sigma})$, $g$ is constant on the set $\sigma^\ell([ab])$.

 Indeed, let $i$ be a positive integer such that for all $x\in X_{\sigma}$, $g$ depends only on $x_{[-i,i]}$.  Such an integer exists since $g$ is  locally constant, as observed in Section~\ref{subsec:coboundaries}. Take $\ell$ large enough so that $\ell \geq k+d$ and $\min\{|\sigma^\ell(a)|:a\in A\}>i$. Since $f$ is constant on the atoms of $\mathcal{P}_k$, so is it on those of $\mathcal{P}_\ell$.
Let $ab\in \mathcal{L}_2(X_{\sigma})$ and $y,z\in \sigma^\ell([ab])$. Since $g(x)$ depends only on $x_{[-i,i]}$, $g\circ T^i(x)$ depends on $x_{[0,2i]}$. Since $y,z\in \sigma^\ell([ab])$ and $|\sigma^\ell(a)|, |\sigma^\ell(b)|>i$, $y$ and $z$ share the same $2i$ first coordinates and  thus $g\circ T^i(y)=g\circ T^i(z)$. On the other hand, for all $0\leq j<|\sigma^\ell(a)|$, $T^j(y)$ and $T^j(z)$ are in the same atom of $\mathcal{P}_\ell$, so in particular for all $0\leq j<i$, $T^j(y)$ and $T^j(z)$ are in the same atom of $\mathcal{P}_\ell$. Since $f$ is constant on the atoms of $\mathcal{P}_\ell$, we obtain that $f^{(i)}(y)=f^{(i)}(z)$, by recalling that
 $f^{(i)} (x)$ stands for  $f(x)+f\circ T(x)+\cdots+f\circ T^j(x)+\cdots+ f\circ T^{i-1}(x).$
Finally, note that for all $x\in X_{\sigma}$ and for all $s\in\mathbb{N}$, $g(x)=g\circ T^s(x)-f^{(s)}(x)$, which implies that $g(y)=g(z)$, which ends the proof of the claim, that is, $g$ is constant on each  atom of the  base.

We thus can define  a map  $\psi\in Z_2(X_{\sigma})$  as $\psi(ab)=g(x)$ for $x\in \sigma^\ell([ab])$. Then, if $x\in \sigma^\ell([ab])$ and $T^{|\sigma^\ell(a)|}(x)\in \sigma^\ell([bc])$, we have
\begin{equation} \label{eq:cob}
\psi(bc)-\psi(ab)= g\circ T^{|\sigma^\ell(a)|}(x)-g(x)=f^{(\mid\sigma^\ell(a)\mid)}(x)=\phi_\ell(ab).
\end{equation}

The function $\phi_n$ defined in \eqref{phin} can be seen as a vector in $\mathbb{R}^{d}$. In the following we refer to it indistinctly as a function belonging to $\mathbb{R}^{\mathcal{L}_2(X_\sigma)}$ or as a vector in $\mathbb{R}^{d}$.

  Let $ n \geq k+d$. We now want to prove  that  by multiplying $\phi_n$ by a suitable power of $M_{\sigma_2}$  yields an element  of  $\beta(R_1(X_{\sigma}))$.

We recall that,  given a $d\times d$-matrix $M$, its { eventual range}  $\mathcal{R}_M$   and  its {eventual kernel} $\mathcal{K}_M$ are respectively 
$$\mathcal{R}_M=\bigcap_{k\geq 1}M^k\mathbb{R}^d, \quad   \quad \mathcal{K}_M=\bigcup_{k\geq 1}\ker(M^k).$$
Note that $\mathbb{R}^d=\mathcal{R}_M\oplus\mathcal{K}_M$,  $\mathcal{R}_M=M^d\mathbb{R}^d$ and  $\mathcal{K}_M=\ker(M^d)$ (see e.g.  \cite[Chapter 7]{LindMarcus:95}).
First observe that $\phi_n\in\mathcal{R}_{M_{\sigma_2}}$. Indeed, by following Remark \ref{rem:m2}, one can show that
$$\phi_n=M_{\sigma_2}^{n-k}\phi_k\in M_{\sigma_2}^{n-k}\mathbb{R}^{\mathcal{L}_2(X_{\sigma})},$$
and since $n-k\geq d$ and $\mathcal{R}_{M_{\sigma_2}}=M_{\sigma_2}^d\mathbb{R}^d$, we conclude that $\phi_n\in \mathcal{R}_{M_{\sigma_2}}$.

Again thanks to Remark \ref{rem:m2}, and by assuming $\ell\geq n$, we obtain that, for every $abc\in \mathcal{L}_3(X_{\sigma})$,   
$$\psi(bc)-\psi(ab)=(M_{\sigma_2}^{\ell-n}\phi_n)(ab).$$

Choose $m$ large enough so that $|\sigma^m(a)|\geq 2$ for every $a\in {\mathcal A}$, and define $\theta\in Z_1(X_{\sigma})$ by $\theta(a)=\psi(a_1a_2)$ for $a\in {\mathcal A}$ if $\sigma^m(a)=a_1\cdots a_r$. If $ab\in \mathcal{L}_2(X_{\sigma})$ with $\sigma^m(a)=a_1\cdots a_r$ and $\sigma^m(b)=b_1\cdots b_s$, we obtain
\begin{eqnarray*}
(M_{\sigma_2}^{\ell-n+m}\phi_n)(ab)&=&(M_{\sigma_2}^{\ell-n}\phi_n)(a_1a_2)+\cdots+(M_{\sigma_2}^{\ell-n}\phi_n)(a_rb_1)\\
&=& \psi(b_1b_2)-\psi(a_1a_2)=\theta(b)-\theta(a)=(\beta\theta)(ab),
\end{eqnarray*}
and it follows that $M_{\sigma_2}^{\ell-n+m}\phi_n$ belongs to $\beta(Z_1(X_{\sigma}))$. 

It remains to prove that  $\phi_n\in \beta(R_1(X_{\sigma}))$. In particular, $M_{\sigma_2}^{\ell-n+m}\phi_n\in \beta(R_1(X_{\sigma}))$. 
Choosing $m$ large enough,  we can assume that $M_{\sigma_2}^{\ell-n+m}\phi_n\in \mathcal{R}_{M_{\sigma_2}}$. Since the subspace $\beta(R_1(X_{\sigma}))$ is invariant under $M_{\sigma_2}$ and $M_{\sigma_2}$ is an automorphism of $\mathcal{R}_{M_{\sigma_2}}$, we obtain that
$$M_{\sigma_2}^{\ell-n+m}: \mathcal{R}_{M_{\sigma_2}}\cap \beta(R_1(X_{\sigma}))\to \mathcal{R}_{M_{\sigma_2}}\cap \beta(R_1(X_{\sigma}))$$
is a bijection. Therefore, there exists a unique $\varphi \in \mathcal{R}_{M_{\sigma_2}}\cap \beta(R_1(X_{\sigma}))$ such that
$$M_{\sigma_2}^{\ell-n+m}\phi_n=M_{\sigma_2}^{\ell-n+m}\varphi,$$
and then $\phi_n=(\phi_n-\varphi)+\varphi$ belongs to $\mathcal{K}_{M_{\sigma_2}}+\mathcal{R}_{M_{\sigma_2}}\cap \beta(R_1(X_{\sigma}))$. Finally, recall that $\mathbb{R}^d=\mathcal{R}_{M_{\sigma_2}}\oplus \mathcal{K}_{M_{\sigma_2}}$ and $\phi_n\in \mathcal{R}_{M_{\sigma_2}}$. This  implies that $\phi_n-\varphi=0$ and thus  $\phi_n\in \mathcal{R}_{M_{\sigma_2}}\cap \beta(R_1(X_{\sigma}))\subseteq \beta(R_1(X_{\sigma}))$.
\end{proof}

We now translate the previous proposition in terms of balancedness for substitutive symbolic systems having rational frequencies. 
\begin{proposition}\label{t1}
Let $\sigma$ be a primitive substitution.  Let $v\in \mathcal{L}_{\sigma}$ having a rational frequency  $\mu_v$ and  $f_v= \chi_{[v]}-\mu_v\in C(X_{\sigma},\mathbb{R})$. There exists 
 $k\in\mathbb{N}$ be  such that $f_v$ is constant in the atoms of  the two-letter partition $\mathcal{P}_k$. 
If  $(X_{\sigma},T)$ is balanced on $v$, then $\phi_{v,b}\in\beta(R_1(X))$ for all $n\geq k+d$, where $d=\Card \mathcal{L}_2(X)$  and $\phi_{v,n}$ is defined as in \eqref{phin}, i.e., 
$$\phi_{v,n}(ab)=\sum_{j=0}^{|\sigma^n(a)|-1} f_v\mid_{T^j\sigma^n([ab])} \quad \forall ab\in \mathcal{L}_2(X_{\sigma}).$$
\end{proposition}
\begin{proof} We write $\mu_v= p_v/q_v$  in  irreducible form.
For all $n\geq 0$, the two-letter  partition $\mathcal{P}_n$ (as defined in \eqref{pn}) verifies that all elements in any atom of $\mathcal{P}_n$ share at least their $L_n+1$ letters, where $L_n=\min\{|\sigma^n(a)|:a\in\mathcal{A}\}$. Therefore, for  all $k$ large, $f_v$ (and consequently $q_v\cdot f_v$) is constant in the atoms of $\mathcal{P}_k$. By Lemma \ref{eq_cob}, since $(X_{\sigma},T)$ is balanced in $v$, $f_v$ is a coboundary, and then so is $q_v\cdot f_v$. By Proposition  \ref{beta_1}, $q_v\cdot \phi_n\in \beta(R_1(X))$ for all $n\geq k+d$, and consequently $\phi_n\in \beta(R_1(X))$ for all $n\geq k+d$.
\end{proof}

We now deduce from  Proposition~\ref{t1} necessary conditions for balancedness.  Let $(X,T)$ be  a minimal symbolic system on the alphabet $\mathcal{A}$ and  let $a\in\mathcal{A}$. 
 We   recall that  a word $w$  with $wa \in {\mathcal L} (X)$  is a {\it return word} to the letter  $a$
 if $a$ is a prefix of $wa$.
 
\begin{lemma}\label{return}
Let $(X,T)$ be  a minimal symbolic system defined  on the alphabet $\mathcal{A}$, $a\in\mathcal{A}$ and $w=w_0\cdots w_{|w|-1}$ be a return word to $a$. 
If $\phi\in \beta(R_1(X))$, then
 $$\phi(w_{|w|-1}a)+\sum_{i=1}^{|w|-1}\phi(w_{i-1}w_i)=0.$$
\end{lemma}  
\begin{proof}
One has $w_0w_1, w_1w_2,\cdots, w_{|w|-2}w_{|w|-1}, w_{|w|-1}a\in \mathcal{L}_2(X)$.
The result  follows  directly  from the definition of return words and from the fact that there exists $\varphi\in R_1(X)$ such that $\phi=\beta\varphi$.
\end{proof}
We now  can prove Theorem  \ref{condnec1}.

\begin{proof}
By  Proposition \ref{t1}, $\phi_{v,n} \in \beta(R_1(X_{\sigma}))$ for all  $n$ large. For any  $ab\in \mathcal{L}_2(X_{\sigma})$ \begin{eqnarray}\label{eqphi}
\phi_n(ab)=\alpha_{ab}\left(1-\frac{p_v}{q_v}\right)-(|\sigma^n(a)|-\alpha_{ab})\cdot  \frac{p_v}{q_v},
\end{eqnarray}
where $$\alpha_{ab}=\Card\{0\leq j<|\sigma^n(a)|: T^j\sigma^n([ab])\subseteq [v]\},$$
that is, $\alpha_{ab}$ is the number of levels in the $ab-$tower of $\mathcal{P}_n$ in which all elements begin with the word $v$.
Using Lemma \ref{return} and  \eqref{eqphi}, we obtain 
\begin{eqnarray*}
0&=&\alpha_{w_{|w|-1}a}\left(q_v-p_v\right)-(|\sigma^n(w_{|w|-1})|-\alpha_{w_{|w|-1}a})\cdot p_v+\\
& & \sum_{i=1}^{|w|-1} \alpha_{w_{i-1}w_i}\left(q_v-p_v\right)-(|\sigma^n(w_{i-1})|-\alpha_{w_{i-1}w_i})\cdot p_v
\end{eqnarray*}

which implies
\begin{eqnarray*}
q_v \left(\alpha_{w_{|w|-1}a}+\sum_{i=1}^{|w|-1} \alpha_{w_{i-1}w_i} \right)&=& p_v\left( |\sigma^n(w_{|w|-1})|+\sum_{i=1}^{|w|-1}|\sigma^n(w_{i-1})|\right)\\
&=& p_v|\sigma^n(w)|.
\end{eqnarray*}
The integers $p_v$ and $q_v$ being coprime, either 
$$\left(\alpha_{w_{|w|-1}a}+\sum_{i=1}^{|w|-1} \alpha_{w_{i-1}w_i} \right)=0$$
or $q_v$ divides $|\sigma^n(w)|$.  Since  $|\sigma^n(w)| \neq 0$,   we conclude that $q_v$ divides $|\sigma^n(w)|$, which ends the proof of the first assertion.

We now prove the second assertion of Theorem~\ref{condnec1}.
Let $a\in\mathcal{A}$ and  assume that  there exist $b,c$ such that $bac$ belongs to $\mathcal{L}_3(X_{\sigma})$ and $bc\in \mathcal{L}_2(X_{\sigma})$. Since $\phi_{v,n}\in \beta(R_1(X_{\sigma}))$ and $ba, ac, bc \in \mathcal{L}_2(X_{\sigma})$, one has $\phi_n(ba)+\phi_n(ac)=\phi_n(bc)$, that is,
\begin{eqnarray*}
0&=&\alpha_{ba}(q_v-p_v)-p_v(|\sigma^n(b)|-\alpha_{ba})+\alpha_{ac}(q_v-p_v)-p_v(|\sigma^n(a)|-\alpha_{ac})\\
& &-\alpha_{bc}(q_v-p_v)+p_v(|\sigma^n(b)|-\alpha_{bc})\\
&=&(\alpha_{ba}+\alpha_{ac}-\alpha_{bc})q_v-p_v|\sigma^n(a)|.
\end{eqnarray*}
The integers $p_v$ and $q_v$ being coprime, either $\alpha_{ba}+\alpha_{ac}-\alpha_{bc}=0$ or $q_v$ divides $|\sigma^n(a)|$. Here again  $\alpha_{ba}+\alpha_{ac}-\alpha_{bc} \neq 0$, since $|\sigma^n(a)|\neq 0$,   hence $q_v$ divides $|\sigma^n(a)|$.
\end{proof}
\begin{remark}\label{rem:example}
Note that Proposition \ref{t1} gives us the smallest $n$ for which the conclusions of both parts of Theorem \ref{condnec1} are always  true. It corresponds to $n=k+d$ and thus it can be determined in an effective way. See Example \ref{ex:ch} 
below for an application.
\end{remark}

As a consequence of the previous theorem, we have the following corollary about the Thue--Morse substitution.

\begin{corollary}[Thue--Morse substitution]\label{cor:TM}
Let $\sigma_{TM}$ be the Thue--Morse substitution on $\{0,1\}$ given by
$\sigma_{TM} \colon 0\mapsto 01$, $1\mapsto 10$.  The subshift $(X_{\sigma_{TM}},T)$ is balanced on   letters  but it  is unbalanced  on any  factor of length $\ell$, with 
 $\ell\geq 2$. \end{corollary}
 \begin{proof}
 From \cite[Theorem1]{Dekking:92},  we know that  the frequency  $\mu_v$ of a factor  $v$ of length  $\ell \geq 2$  verifies
$\mu_v=\frac{1}{6}2^{-m}$ or $\mu_v=\frac{1}{3}2^{-m}$, 
where $m$ is such that $2^m<\ell\leq 2^{m+1}$. Frequencies are then rational, $p_v=1$,  and $q_v\in \{3\cdot 2^{m+1}, 3\cdot 2^{m}\}$.
Note that $00$   belongs  to $ \mathcal{L}_2(X_{\sigma_{TM}}$). We  then apply the first condition of Theorem~\ref{condnec1}.
 \end{proof}

We  also deduce from Theorem~\ref{condnec1} the following.
\begin{corollary} \label{cor:sym}
Let $\sigma$ be   primitive substitution of  constant length $\ell$  over the alphabet ${\mathcal A}$ of cardinality $d$  such that its 
substitution matrix is symmetric and  $d$     is coprime with  $\ell$, or 
does  not divide $\ell ^n$, for  all $n$ large. If there exists a letter $a$ and  a return word 
$w$ to $a$ such that $d$ is coprime with $|w|$.   Then,
$(X_{\sigma},T)$ is not balanced on letters.
\end{corollary}
\begin{proof}
 The  substitution matrix $M_{\sigma}$  admits as   left eigenvector (and thus as right eigenvector) 
associated with  the eigenvalue $\ell$ the  vector  with coordinates  all equal  to $1$.
One thus has $\mu_a=1/d$ for all $a$ and we apply   the first part of Theorem~\ref{condnec1}.
\end{proof}

\subsection{Examples}\label{subsec:examples}

\begin{example}[Chacon substitution]\label{ex:ch}
The primitive Chacon substitution $\sigma_C$  is defined over the
alphabet $\{1,2,3\}$  by
$\sigma_C \colon  1 \mapsto 1123, 2 \mapsto 23,  3 \mapsto 123.$
The spectrum of $M_{\sigma_C}$ is $\{3,1,0\}$. We cannot  apply directly Theorem~\ref{theo:Adam}.
The letter frequency vector is   $(1/3,1/3,1/3)$  and then $q_1=q_2=q_3=3$. One has $11 \in {\mathcal L}_2(X_{\sigma_C})$, and then, for every $a\in\{1,2,3\}$, if the system is balanced on $a$, $3$ divides $|\sigma_C^n(1)|$ for all $n\geq k+d$ (see Proposition \ref{t1} for notation and  Remark \ref{rem:example}). In this case, it is enough to take $k=1$;  moreover  one has  $d=5$, so that $3$ divides $|\sigma_C^6(1)|$.  But $|\sigma_C^6(1)|=1093$, which is not divisible by $3$. We conclude that $(X_{\sigma_C},T)$ is neither balanced on   letters,   nor balanced on factors of   a given size,  by Proposition \ref{prop:decrease}. In view of  Theorem \ref{theo:equilibre}, this is consistent with the fact that
$ (X_{\sigma_C}, T)$ is weakly  mixing, that is, it admits no non-trivial  topological eigenfunction \cite[Lemma 5.5.1]{Fog02}.  \end{example}

\begin{example}[Toeplitz substitution]
The Toeplitz  substitution,  
also called Period doubling substitution, is defined   over $\{0,1\}$ as 
$\sigma_T\colon 0 \mapsto 01$, $ 1 \mapsto 00.$ The 
spectrum of $M_{\sigma_T}$ is  $\{2,-1\}$. We cannot  apply directly Theorem~\ref{theo:Adam}.
The frequencies of letters  are $1/3$ and $2/3$ and then $q_0=q_1=3$. One has $00 \in {\mathcal L}_2(X_{\sigma_C})$, so if the system is balanced on $a\in\{0,1\}$, $3$ divides $|\sigma_T^n(0)|$ for all $n$ large. Since for all $n\geq 1$, $|\sigma_T^n(0)|=2^n$, which is not divisible by 3, $(X_{\sigma_T},T)$ is neither balanced on  letters,  by applying  the first condition of Theorem~\ref{condnec1},  nor   balanced on factors of a given size, by Proposition \ref{prop:decrease}.
\end{example}

\begin{example} Consider the substitution 
$\sigma \colon 0 \mapsto 11,$ $   1 \mapsto 21$,  $ 2 \mapsto 10$.
The spectrum of  $M_{\sigma}$ is  $\{2, -1/2 (\sqrt 3 i +1),  1/2 (\sqrt 3 i -1)\}$.
Once again, we cannot  apply directly Theorem~\ref{theo:Adam}.
The frequencies of letters  are 
$(1/7, 4/7,2/7)$ and $11$ is in its language.   We  apply the first condition of Theorem~\ref{condnec1} to deduce that   $(X_{\sigma_T},T)$ is not balanced.

\end{example}

\begin{example}
Consider the substitution $\sigma \colon 0 \mapsto 010,$  $ 1 \mapsto 102$, $ 2 \mapsto 201$ from \cite[Example 6.2]{Queffelec:2010}. The spectrum of $M_\sigma$ is $\{3,1,0\}$.
We can neither  apply  Theorem~\ref{theo:Adam}, nor     the criteria of Theorem~\ref{condnec1}.  The letter frequency vector is $(1/2, 1/3, 1/6)$ and  $q_0=2$, $q_1=3$, $q_2=6$. 
No  factor of the form $aa$ appears in ${\mathcal L} (X_{\sigma})$. For instance, the word $w=01$ is a return word to $0$,
and $|\sigma^n(w)|=2\cdot 3^{n}$ for all $n$, which provides  no contradiction with   $q_a$ divides $|\sigma^n(w)|$.   In fact the substitution $\sigma$
is balanced on  letters. Indeed, consider the substitution $\tau$ on the alphabet  $\{a,b\}$
defined by $\tau \colon  a \mapsto aab$, $b \mapsto aba$,  and the morphism $\varphi$ from
$\{a,b\}^*$ to $\{0,1,2\}^*$  defined by   $\varphi(a)=01$ and $\varphi(b )=02$.
One has $\sigma \circ \varphi= \varphi \circ \sigma$ which implies that  
$\sigma^{\infty} (0)=  \varphi ( \tau ^{\infty} (a))$.
The substitution $\tau $ is balanced on letters since the spectrum   of $M_\tau$ is $\{2,0\}$ which implies that
$\sigma$ is balanced on letters.   \end{example}

\begin{example} 
Consider the  substitution 
$\sigma \colon 0 \mapsto 001, \ 1 \mapsto 101 $.
The spectrum of  its substitution matrix is $\{1,3\}$. 
The frequencies of its letters  are $\mu_0=\mu_1=1/2$ and $00$ is   a factor. We deduce from Corollary \ref{cor:sym}
that it is not balanced  on   letters.
\end{example}

\footnotesize{

\bibliographystyle{alpha}
\bibliography{Balance}}
\end{document}